\documentclass[11pt,a4paper]{article}
\usepackage{fullpage}
\usepackage{times}

\usepackage{enumitem}
\usepackage{bbm}
\usepackage{amssymb,amsthm,amsmath,amssymb,wrapfig,dsfont,mathtools}
\usepackage{thmtools}

\usepackage{booktabs}

\usepackage{algorithmic}
\usepackage[linesnumbered,ruled,vlined]{algorithm2e}
\usepackage[capitalize]{cleveref}
\title{Best of Both Worlds: Agents with Entitlements\thanks{Institute for Computer Science, Goethe University Frankfurt, Robert-Mayer-Strasse 11-15, 60325 Frankfurt am Main, Germany. \texttt{$\{$mhoefer,schmalhofer,varricchio$\}$@em.uni-frankfurt.de}. Martin Hoefer and Giovanna Varricchio were supported by DFG grant Ho 3831/5-1.}}
\author{Martin Hoefer \and Marco Schmalhofer \and Giovanna Varricchio}
\date{}

\newcommand{\agents}{\mathcal{N}}
\newcommand{\goods}{\mathcal{G}}
\newcommand{\allocation}{\mathcal{A}}
\newcommand{\expectation}[1]{\mathbb{E}\left[#1\right]}

\newcommand{\set}[1]{\{#1\}}
\newcommand{\ceil}[1]{\lceil#1\rceil}
\newcommand{\floor}[1]{\lfloor#1\rfloor}
\newcommand{\modulus}[1]{\vert #1 \rvert }

\newcommand{\lottery}{\mathcal{L}}

\newcommand\scalemath[2]{\scalebox{#1}{\mbox{\ensuremath{\displaystyle #2}}}}

\newcommand{\UG}{\mathcal{H}^{\scalemath{0.6}{\mathsf{UG}}}}

\newcommand{\instance}{\mathcal{I}}

\newcommand{\RR}{\ensuremath{\mathbb{R}}}

\newcommand{\WEF}{\mathsf{WEF}}
\newcommand{\WWEF}{\mathsf{WWEF}}
\newcommand{\EF}{\mathsf{EF}}
\newcommand{\WPROP}{\mathsf{WPROP}}
\newcommand{\PROP}{\mathsf{PROP}}
\newcommand{\WGF}{\mathsf{WGF}}
\newcommand{\SD}{\mathsf{SD}}
\newcommand{\WSD}{\mathsf{WSD}}
\newcommand{\dominates}{\succeq^{\scalemath{0.6}{\SD}}}
\newcommand{\weakdominates}{\succ^{\scalemath{0.6}{\SD}}}
\newcommand{\DSE}{\textsc{DSE}}
\newcommand{\XDSE}{X^{\scalemath{0.6}{\DSE}}}

\newcommand{\Eaten}{\mathsf{Eaten}}

\newtheorem{proposition}{Proposition}

\newtheorem{theorem}{Theorem}
\newtheorem{lemma}{Lemma}
\newtheorem{corollary}{Corollary}

\theoremstyle{definition}
\newtheorem{example}{Example}
\newtheorem{definition}{Definition}

\newcommand{\exampleend}{\hfill$\small\blacksquare $}

\DeclareMathOperator*{\argmax}{arg\,max}

\declaretheoremstyle[
spaceabove=\topsep, spacebelow=\topsep,
headfont=\normalfont\bfseries,
notefont=\bfseries, notebraces={}{},
bodyfont=\normalfont\itshape,
postheadspace=0.5em,
name={\ignorespaces},
numbered=no,
headpunct=.]
{mystyle}

\declaretheorem[style=mystyle]{customtheorem}

\begin{document}

\maketitle

\begin{abstract}
Fair division of indivisible goods is a central challenge in artificial intelligence. For many prominent fairness criteria including {\em envy-freeness} ($\EF$) or {\em proportionality} ($\PROP$), no allocations satisfying these criteria might exist. Two popular remedies to this problem are randomization or relaxation of fairness concepts. A timely research direction is to combine the advantages of both, commonly referred to as \emph{Best of Both Worlds} (BoBW). 

We consider fair division \emph{with entitlements}, which allows to adjust notions of fairness to heterogeneous priorities among agents. This is an important generalization to standard fair division models and is not well-understood in terms of BoBW results. Our main result is a lottery for additive valuations and different entitlements that is ex-ante \emph{weighted envy-free} ($\WEF$), as well as ex-post \emph{weighted proportional up to one good} ($\WPROP1$) and \emph{weighted transfer envy-free up to one good} ($\WEF(1,1)$). It can be computed in strongly polynomial time. We show that this result is tight -- ex-ante $\WEF$ is incompatible with any stronger ex-post $\WEF$ relaxation.

In addition, we extend BoBW results on group fairness to entitlements and explore generalizations of our results to instances with more expressive valuation functions.
\end{abstract}


\section{Introduction}
Fair division of a set of indivisible goods is a prominent challenge at the intersection of economics and computer science. It has attracted a lot of attention over the last decades due to many applications in both simple and complex real-world scenarios. Formally, we face an allocation problem with finite sets $\agents$ of $n$ agents and $\goods$ of $m$ goods. Each agent $i \in \agents$ has a valuation function $v_i : 2^{\goods} \rightarrow \mathbb{R}_{\geq 0}$. The goal is to compute a ``fair'' allocation $\allocation = (A_1,\ldots,A_n)$, i.e., a fair partition of the goods among the agents. 

What is fair can certainly be a matter of debate. For this reason, several fairness criteria have been introduced and studied. {\em Envy-freeness} ($\EF$) is probably one of the most intuitive concepts -- it postulates that once goods are allocated no agent strictly prefers goods received by any other agent, i.e., $v_i(A_i) \geq v_i(A_j)$ for all $i,j\in \agents$. $\EF$ is a comparison-based notion. In contrast, there are also threshold-based ones such as {\em proportionality} ($\PROP$): $\allocation$ is proportional if every agent receives a bundle whose value is at least her proportional share, i.e., $v_i(A_i)\ge v_i(\goods)/n$ for every $i\in \agents$.

Unfortunately, for indivisible goods, neither $\PROP$- nor $\EF$-allocations may exist. Two natural conceptual remedies to this non-existence problem are (1) randomization or (2) relaxation of fairness concepts. Towards (1), a random allocation that is $\EF$ in expectation always exists (for every set of valuation functions): Select an agent uniformly at random and give the entire set of goods $\goods$ to her. Then, however, every realization in the support is highly unfair -- there is always an agent who receives everything, while all others get nothing. Moreover, it is easy to see that such an allocation might not even be Pareto-optimal. Towards (2), a well-known relaxation of $\EF$ is {\em envy-freeness up to one good} ($\EF1$)~\cite{liptonMMS04, Budish11}: Every agent shall value her own bundle at least as much as any other agent's bundle after removing some good from the latter, i.e., for every $i,j\in \agents$ there is $g \in \goods$ such that $v_i(A_i) \geq v_i(A_j \setminus \{g\})$. Whenever the valuations of the agents are monotone, an $\EF1$ allocation always exists and can be computed in polynomial time~\cite{liptonMMS04}. However, different $\EF1$ allocations may advantage different agents. Similarly to $\EF1$, {\em proportionality up to one good} ($\PROP1$) has also been studied~\cite{conitzerFS17}. 

A timely research direction is to combine the advantages of both randomization and relaxation, commonly referred to as {\em Best of Both Worlds} (BoBW) results. An important result was obtained by both Aziz~\cite{aziz20} and Freeman et al.~\cite{freemanSV20} for additive valuations -- a lottery that is $\EF$ in expectation (ex-ante) and $\EF1$ for every allocation in the support (ex-post). Moreover, it can be implemented as a lottery over deterministic allocations with polynomial-sized support. Both papers generalize the Probabilistic Serial (PS) rule~\cite{BogomolnaiaM01} for the matching case, when there are $n$ agents and $m = n$ goods. PS is ex-ante $\EF$. By the Birkhoff-von Neumann decomposition, it can be represented as a lottery over polynomially many deterministic allocations. Furthermore, any allocation in the support assigns to each agent exactly one good. This implies ex-post $\EF1$. Both \cite{aziz20,freemanSV20} generalize the application of the Birkhoff-von Neumann decomposition to instances with arbitrarily many goods.

In our work, we consider a more general framework to allow more flexibility in the definition of fairness. Concepts like $\EF$ or $\PROP$ imply that all agents are symmetric, i.e., they are ideally treated as equals. In many scenarios, however, there is an inherent asymmetry in the agent population. Alternatively, it can be beneficial for an allocation mechanism to have the option to reward certain agents. We follow the formal framework of \emph{entitlements}~\cite{chakrabortyIS21, azizMS20} that enables increased expressiveness. Formally, each agent $i \in \agents$ now has a weight or a priority $w_i > 0$. Fairness notions like $\EF$ or $\PROP$ are then refined based on the weights (see Section~\ref{sec:preliminaries} for formal definitions). Generally, we will use a prefix ``$\textsf{W}$'' to refer to a fairness concept in the context of entitlements.

\subsection{Our Contribution}
We study lotteries that guarantee both ex-ante and ex-post guarantees for fair division with additive valuations and different entitlements. We provide a lottery that is ex-ante {\em weighted stochastic-dominance envy-free} ($\WSD$-$\EF$) and consequently ex-ante $\WEF$. Applying a bihierarchy decomposition by Budish et al.~\cite{budishCK13}, we show it is possible to achieve ex-post $\WPROP1$ and $\WEF(1,1)$. The latter implies that in every allocation $\allocation$ in the support, weighted envy from agent $i$ to $j$ can be eliminated by moving entirely one good from $A_j$ to $A_i$. All our constructions can be carried out in strongly polynomial time. Perhaps surprisingly, this result is tight -- we show that ex-ante $\WEF$ is incompatible with any stronger ex-post $\WEF$ notion. In particular, a direct extension of \cite{aziz20,freemanSV20} to a lottery with ex-ante $\WEF$ and ex-post $\WEF1$ is impossible.

Freeman et al.~\cite{freemanSV20} investigate further combinations of ex-ante and ex-post properties; namely, they provide a lottery that is ex-ante {\em group fair} as well as ex-post $\PROP1$ and $\EF_1^1$. In an $\EF_1^1$-allocation $\allocation$, we can eliminate envy from $i$ to $j$ when we remove one good from $A_j$ and add one good to $A_i$; differently from $\EF(1,1)$, the good added to $A_i$ is \emph{not required} to come from $A_j$. We prove that this result can be adapted to hold also for instances with entitlements. 

Finally, we expand the scope of BoBW 
towards more general valuations. For equal entitlements, ex-ante $\EF$ and ex-post $\EF1$ is possible in more general cases. For different entitlements, ex-ante $\WEF$ and ex-post $\WEF(1,1)$ or $\WPROP1$ are no longer compatible (even for two agents, one additive and one unit-demand). For this reason, we focus on threshold-based guarantees -- we show that it is possible to compute in polynomial time a lottery that is ex-ante $\WPROP$ and ex-post $\WPROP1$, even for XOS valuations.

\subsection{Related Work}
Fair division attracted an enormous amount of attention, and there is a large number of surveys. We refer to a rather recent one by Amanatidis et al.~\cite{amanatidisBFV22} and restrict attention to more directly related works.

In the context of BoBW, the $\mathsf{MMS}$-share is studied by Babaioff et al.~\cite{babaioffEF21}, who design a lottery simultaneously achieving ex-ante $\PROP$ and ex-post $\PROP1$ + $\frac{1}{2}$-$\mathsf{MMS}$.

An orthogonal direction is pursued in~\cite{caragiannisKK21} by introducing interim $\EF$, a trade-off between ex-ante and ex-post $\EF$. Interim $\EF$ requires that agent $i$, knowing the realization of her bundle $A_i$, obtains more value from $A_i$ than the (conditional) expected value of the bundle of any other agent $j$. Formally, $v_i(S)\geq \expectation{v_i(A_j) \vert A_i=S}$ for each $S\subseteq \goods$ such that $\mathbb{P}(A_i=S)>0$. Interim-$\EF$ implies ex-ante $\EF$, and ex-post $\EF$ implies interim $\EF$. Furthermore, any interim $\EF$ allocation is ex-post $\PROP$, which implies that interim $\EF$ allocations may not exist. In the matching case, a cleverly crafted separation oracle for a linear program and an extension of the Birkhoff-von Neumann decomposition can be combined to design a polynomial-time algorithm that either provides an interim $\EF$ allocation or determines it does not exist~\cite{caragiannisKK21}.

When agents are endowed with ordinal preferences rather than cardinal valuation functions, stochastic-dominance envy-freeness is the most prominent fairness notion for lotteries. It was first considered in~\cite{BogomolnaiaM01} and later systematically studied in \cite{AzizGMW14}.

For fair division with entitlements, the literature has focused on characterizing picking sequences guaranteeing fairness properties~\cite{chakrabortyIS21, chakrabortySKS21, chakrabortySHS22}, the problem of maximizing Nash social welfare~\cite{GargKK20,suksompongT22}, and the definition of appropriate shares~\cite{babaioffEF21share}. To the best of our knowledge,  BoBW results have not been addressed so far.


\section{Preliminaries}
\label{sec:preliminaries}

A fair division instance $\instance$ is given by a triple $(\agents, \goods, \set{v_i}_{i\in \agents})$, where $\agents$ is a set of $n$ agents and $\goods$ is a set of $m$ indivisible goods. Every agent $i\in \agents$ has a valuation function $v_i: 2^\goods \rightarrow \mathbb{R}^+$, where $v_i(A)$ represents the value, or utility, of $i$ for the bundle $A\subseteq \goods$. We assume that valuations are monotone ($v(A) \le v(B)$ for $A \subseteq B$) and normalized $(v(\emptyset) = 0$). 

\paragraph{Classes of Valuations.}
For each $i\in \agents$ and $g\in \goods$, $v_{i}(g) \ge 0$ represents the value $i$ assigns to the good $g$. A valuation $v_i$ is \emph{additive} if $v_i(A) = \sum_{g \in A} v_{i}(g)$.

A valuation function $v_i$ is \emph{multi-demand} if there exists $k\in\mathbb{N}$ such that $v_i(A)$ is given by the sum of the $k$ most valuable goods in $A$ for $i$. If $k=1$ we talk about \emph{unit-demand}. 

A valuation function $v_i$ is  \emph{cancelable} if $v_i(S \cup \{g\}) > v_i(T \cup \{g\}) \implies v_i(S) > v_i(T) $,  for all $S, T \subseteq \goods$ and $\forall g \in \goods \setminus (S \cup T)$. Cancelable valuations generalize several classes studied in the literature, e.g., additive, weakly-additive, budget-additive, product, and unit-demand.

A valuation function $v_i$ is \emph{XOS} if there is a family of additive set functions $\mathcal{F}_i$ such that $v_i(A)= \max_{f\in \mathcal{F}_i}f(A)$.  
XOS  generalize additive and submodular valuations.
\bigskip

In what follows, whenever we sort the goods in $\goods$ according to the valuation function of a specific agent, we assume ties are broken according to a fixed ordering of $\goods$. This serves to avoid technical and tedious tie-breaking issues.

\paragraph{Entitlements.} 
We study fair division with entitlements. Each agent $i\in \agents$ is endowed with an \emph{entitlement} or \emph{weight}, $w_i > 0$. For convenience, we assume w.l.o.g.\ $\sum_{i\in \agents} w_i = 1$. We say that agents have \emph{equal entitlements} if $w_i=\frac{1}{n}$, for all $i\in \agents$, and refer to this as the \emph{unweighted setting}.

We next give an example of fair division instance with entitled and additive agents.
\begin{example}[A fair division instance with entitlements]\label{example:valuations}
	We outline an instance $\instance^*$ given by $(\agents, \goods, \set{v_i}_{i\in \agents})$ and entitlements $w$. The agents are $\agents=\set{1,2,3}$, the goods $\goods=\set{g_1, g_2, g_3, g_4,g_5}$, and $w_1= \frac{1}{2}$, $w_2= \frac{1}{3}$ and $w_3= \frac{1}{6}$ are the entitlements of agent $1, 2$ and $3$, respectively. The valuation functions are additive with values of the agents for single goods shown in \Cref{table:exampleValuations}.
	
	Throughout the paper, whenever we use $\instance^*$, we mean the instance we just described.
	\begin{table}[h!]
		\centering
		\caption{  Agents' valuations in \Cref{example:valuations}}
		\label{table:exampleValuations}
		\begin{tabular}{ccccc}
			\toprule
			$g$ & $g_1$ & $g_2$ & $g_3$ & $g_4$ 
			\\ \midrule         
			$v_1(g)$ & 8 & 8 & 5 & 2
			\\
			$v_2(g)$ & 3 & 5 & 4 &1 
			\\
			$v_3(g)$ & 4 & 7 & 6 &  2
			\\ \bottomrule
		\end{tabular}
	\end{table} 
	\exampleend
\end{example}

\subsection{Weighted Fairness Notions}
An allocation $\allocation=(A_1, \dots, A_n)$ is a disjoint partition of $\goods$ among the agents, where $A_i \cap A_j = \emptyset$, for each $i \neq j$, and $\bigcup_{i\in \agents} A_i = \goods$. An allocation $\allocation$ is \emph{weighted envy-free} ($\WEF$) if $w_j\cdot v_i(A_i) \geq w_i\cdot v_i(A_j)$ for every $i, j\in \agents$. An allocation $\allocation$ is \emph{weighted proportional} ($\WPROP$) if $v_i(A_i)\geq w_i\cdot v_i(\goods)$ holds, for each $i\in \agents$.

Since goods are indivisible such allocations may not always exist, and relaxed versions have been defined. An allocation $\allocation$ is \emph{weighted proportional up to one good} ($\WPROP1$) if for each $i\in \agents$ there exists $g\in \goods \setminus A_i$ such that $v_i(A_i\cup \set{g})\geq  w_i \cdot v_i(\goods)$. Note that, for additive valuations, $\WEF \Rightarrow \WPROP$ but, differently from equal entitlements, $\WEF1 \not\Rightarrow \WPROP1$.

Concerning envy-freeness, we have already discussed $\EF$ and $\EF1$ in the introduction. We here work with a broader definition that generalizes these notions.

\begin{definition}[$\WEF(x,y)$]
For $x,y\in [0,1]$, an allocation $\allocation$ is called $\WEF(x,y)$ if for each $i, j\in \agents$ there exists $g\in A_j$ 
such that 
$w_j\cdot \left(v_i(A_i) + y\cdot v_i(g)\right)\geq w_i\cdot \left(v_i(A_j) -x\cdot v_i(g)\right)$.
\end{definition}

The definition of $\WEF(x,y)$ is meaningful mostly for additive valuations. For general valuations, the idea of $\WEF(1,1)$ can be expressed by $w_j\cdot v_i(A_i \cup \set{g})\geq w_i\cdot v_i(A_j\setminus\set{g})$, and analogously for $\WEF(0,1)$ and $\WEF(1,0)$. Conceptually, $\WEF(1,0)$ coincides with a notion of \emph{weighted envy-freeness up to one good} ($\WEF1$) while $\WEF(0,1)$ with a notion termed \emph{weak weighted envy-freeness up to one good} ($\WWEF1$). $\WEF(1,1)$ has also been called \emph{transfer envy-freeness up to one good} in the unweighted setting. Note that in $\WEF(1,1)$ the good $g$ added to $A_i$ must come from $A_j$. Assuming that $g$ may come from any other bundle leads to the following (weaker) notion.

\begin{definition}[$\WEF_1^1$]
An allocation $\allocation$ is called \emph{weighted envy-free up to one good more and less ($\WEF_1^1$)} if for each $i, j\in \agents$ there exist $g_i, g_j\in \goods$ such that $w_j\cdot v_i(A_i \cup\set{g_i})\geq w_i\cdot v_i(A_j\setminus \set{g_j})$. 
\end{definition}

We move on to fairness concepts for fractional allocations. 
A \emph{fractional allocation} $X = (x_{ig})_{i \in \agents, g\in \goods} \in [0,1]^{n \times m}$ specifies the fraction of good $g$ that agent $i$ receives. We assume fractional allocations are complete, i.e., $\sum_{i\in \agents}x_{ig} =1$ for every $g\in \goods$.

Towards extending group fairness~\cite{ConitzerF0V19} to weighted agents, consider a subset of agents $S\subseteq \agents$. We define $w_S= \sum_{i\in S} w_i$ as the weight of the set, and $\cup_{j \in S} X_j = \left(\sum_{j \in S} x_{jg}\right)_{g \in \goods}$ as the total fraction of each good $g \in \goods$ assigned to the agents of $S$.

\begin{definition}[$\WGF$]
A fractional allocation $X$ is \emph{weighted group fair ($\WGF$)} if for all non-empty subsets of agents $S, T \subseteq \agents$, there is no fractional allocation $X'$ of $\cup_{j\in T} X_j$ to the agents in $S$ such that $w_S \cdot v_i(X'_i) \geq w_T \cdot v_i(X_i)$, for all $i\in S$ and at least one inequality is strict.
\end{definition}

Similarly to the unweighted setting, weighted group fairness implies other (weighted) envy and efficiency notions, for example, $\WEF$ (if $|S|=|T|=1$), $\WPROP$ (if $|S|=1, T=\agents$), and Pareto-optimality (if $S=T=\agents$).

We finally focus on stochastic dominance, a standard fairness notion for random allocations. For convenience, we here define it using fractional allocations. Given any $i\in \agents$, let us denote by $X_i$ and $X_i'$ the fractional bundles of agent $i$ in the allocations $X$ and $X'$, respectively. Agent $i$ $\SD$ prefers $X_i$ to $X'_i$, written $X_i\dominates_i X_i'$, if for any $g^*\in\goods$
\[
	\sum_{g\in\set{g \vert  v_i(g) \geq v_i(g^*) }}  x_{ig} \geq \sum_{g\in\set{g \vert  v_i(g) \geq v_i(g^*) }} x'_{ig} \enspace ,
\]
where $x_{ig}$ and $x'_{ig}$ represents the fraction for $i$ of receiving $g$ in the two random allocations.

We say $X_i\weakdominates_i X_i'$, if $X_i\dominates_i X'_i$ and not $X'_i \dominates_i X_i$. Notice that $\set{g \vert  v_i(g) \geq v_i(g^*) } $ is the set of goods that $i$ likes at least as much as $g^*$. Although we defined it by means of $v_i$, this set only depends on the relative ordering of the goods and not on the valuation $v_i$.  

\begin{definition}[$\SD$-$\EF$ and $\WSD$-$\EF$]
A random allocation $X$ is \emph{$\SD$-envy-free ($\SD$-$\EF$)} if for all $i,j\in\agents$, $X_i\dominates_i X_j$. Similarly, we say $X$ is \emph{$\WSD$-envy-free ($\WSD$-$\EF$)} if for all $i,j\in\agents$, $w_j\cdot X_i\dominates_i w_i\cdot X_j$.
\end{definition}

\subsection{Deterministic Algorithms and Picking Sequences}

For additive valuations, a straightforward round-robin algorithm yields an $\EF1$ allocation. Clearly, when agents have different entitlements, the round-robin algorithm might no longer provide a $\WEF1$ allocation. Different entitlements impose different priorities among agents, which has resulted in the consideration of picking sequences.

A {\em picking sequence} for $n$ agents and $m$ goods is a sequence $\pi=(i_1, \dots, i_m)$, where $i_h\in \agents$, for $h=1,\dots, m$. An allocation $\allocation$ is the result of the picking sequence $\pi$ if it is the output of the following procedure: Initially every bundle is empty; then, at time step $h$, $i_h$ inserts in her bundle the most preferred good among the available ones. Once a good is selected, it is removed from the set of the available goods.
 
In \cite{chakrabortySKS21,chakrabortySHS22}, the authors characterize picking sequences for several fairness criteria, including $\WEF1$, $\WWEF1$ as well as $\WEF(x,1-x)$. 
For our purposes, we will rely on the following characterization for $\WEF(x,y)$ (in the context of additive valuations).

\begin{proposition}\label{thm:pickingWEFxy}
Let $t_i, t_j$ be the number of picks of agents $i$, $j$, respectively, in a prefix of $\pi$. A picking sequence $\pi$ is $\WEF(x,y)$ if and only if for every prefix of $\pi$ and every pair of agents $i,j$, we have $\frac{t_i + y}{w_i}\geq \frac{t_j-x}{w_j}$.
\end{proposition}

Chakraborty et al.~\cite{chakrabortySHS22} prove this proposition using the assumption $x+y=1$, since $\WEF(x,y)$ allocations might not exist for $x+y<1$. The proof can be easily extended, see the Appendix, beyond the assumption to show the statement for all $x,y \in [0,1]$.

Note further that round-robin is not the only picking sequence achieving $\EF1$ for equal entitlements. Any picking sequence that is recursively balanced (RB), i.e.\ $\modulus{t_i-t_j}\leq 1$ in any prefix of $\pi$, results in an $\EF1$ allocation \cite{aziz20}.
 

\subsection{Random Allocations}

A random allocation is a probability distribution $\lottery$ over deterministic allocations. We mostly focus on additive valuations, so we conveniently use a representation as matrix $X$ of marginal assignment probabilities for each good to each agent (i.e., a complete fractional allocation as defined above). We denote by $X^\lottery$ the fractional allocation corresponding to a lottery $\lottery$. Notice that different lotteries might produce the very same fractional allocation.


Throughout the paper, we denote by $X$ (resp.\ $Y$) fractional (resp.\ integral) allocations in matrix form. Further, $X_i$ (resp.\ $Y_i$) denotes a fractional (resp.\ integral) bundle of $i$.

\paragraph{Decomposing Fractional Matrices.}
A decomposition of a fractional matrix $X$ is a convex combination of deterministic allocations, i.e.\ $X= \lambda_1 Y^1 + \dots +\lambda_k Y^k$, where $\sum_{h=1}^k\lambda_h=1$ and $y^h_{ig}\in\set{0,1}$ for each $i\in\agents$, $g\in\goods$ and $h\in[k] = \set{1, \dots, k}$.

A {\em constraint structure} $\mathcal{H}$ consists of a collection of subsets $S \in\agents\times\goods$. Every $S\in \mathcal{H}$ comes with a lower and upper quota denoted by $\underline{q}_S$ and $\overline{q}_S$, respectively. Quotas are integer numbers stored in $\mathbf{q}$.

A $n\times m$ matrix $Y$ is feasible under $\mathbf{q}$ if for each $S\in \mathcal{H}$
\begin{equation*}
\underline{q}_S\leq  \sum_{(i,g)\in S} y_{ig} \leq\overline{q}_S \enspace .
\end{equation*}

A constraint structure $\mathcal{H}$ is a \emph{hierarchy} if, for every $S,S' \in \mathcal{H}$, either $S \cap S'= \emptyset$ or one is contained in the other. $\mathcal{H}$ is a \emph{bihierarchy} if it can be partitioned into $\mathcal{H} = \mathcal{H}_1 \cup \mathcal{H}_2$, such that $\mathcal{H}_1 \cap \mathcal{H}_2 = \emptyset$ and both $\mathcal{H}_1$ and $\mathcal{H}_2$ are hierarchies.

Budish et al.~\cite{budishCK13} generalize the well-known decomposition theorem by Birkhoff and von Neumann:

\begin{theorem}
	\label{thm:budishDecomp}
	Given any fractional allocation $X$,  a bihierarchy $\mathcal{H}$ and corresponding quotas $\mathbf{q}$, if $X$ is feasible under $\mathbf{q}$,  then, there exists a polynomial decomposition into integral matrices. Every matrix in the decomposition is feasible under $\mathbf{q}$. Further, the decomposition can be obtained in strongly polynomial time.
\end{theorem}

In the rest of this paper, given a fractional allocation $X$ and a bihierarchy $\mathcal{H}$, we define the quotas in $\mathbf{q}$ as follows: for every $S\in\mathcal{H}$ we set $\underline{q}_S= \floor{x_S}$ and $\overline{q}_S= \ceil{x_S}$, where $x_S= \sum_{(i,g)\in S}x_{ig}$. The decomposition obtained with these quotas and bihierarchy  $\mathcal{H}$ will be called the $\mathcal{H}$-decomposition.


\paragraph{Utility Guarantee Bihierarchy.} 
An extremely useful bihierarchy is defined as follows. We set $\mathcal{H}_1= \set{C_g\ \vert \ g\in \goods}$, where $C_g=\set{(i,g) \ \vert \ i\in \agents}$ represent columns. Roughly speaking, the hierarchy $\mathcal{H}_1$ ensures that, in any allocation of the decomposition, every good is integrally assigned (and therefore the allocation is complete).

For agent $i\in \agents$, we consider the goods in non-increasing order of $i$'s valuation, i.e., $v_i(g_1) \geq \ldots \geq v_i(g_m)$. Recall that ties are broken according to a predefined ordering of $\goods$. We set $\mathcal{S}_i = \set{\set{g_1}, \set{g_1, g_2}, \dots, \set{g_1, \dots, g_m}}$. In other words, for every $h\in [m]$, $\mathcal{S}_i$ contains a set of the $h$ most preferred goods of $i$. Define $\mathcal{H}_2= \set{(i, S)\ \vert \ i\in \agents, S\in \mathcal{S}_i} \cup \set{(i,g)\vert i\in \agents, g\in \goods}$. The second set of constraints implies that if  $x_{ig}=0$ (resp. $x_{ig}= 1$) then $y_{ig}=0$ (resp. $y_{ig}=1$), for any $Y$ in the decomposition.
Note that (for convenience later on) we slightly abuse notation for $\mathcal{H}_2$ as it is not a set of (row, col)-pairs.

Finally, the \emph{utility guarantee bihierarchy} is given by $\UG= \mathcal{H}_1 \cup \mathcal{H}_2$. Clearly, both $\mathcal{H}_1$ and $\mathcal{H}_2$ are hierarchies.

This bihierarchy was fundamental in~\cite{budishCK13}. We here state it in a slightly stronger version (see~\cite{freemanSV20} for the proof).

\begin{corollary}[Utility Guarantee up to one Good More or Less] \label{cor:UGupToOneGood}
Suppose we are given a fractional allocation $X$, $v_i$ is an additive valuation function, and $v_i(X_i) = \expectation{v_i(X_i)}$. Then for any matrix $Y$ in the $\UG$-decomposition of $X$ the following hold:
\begin{enumerate}
\item if $v_i(Y_i ) < v_i(X_i)$, then $\exists\ g\not\in Y_i$ with $x_{ig}>0$ such that $v_i(Y_i) + v_i(g) > v_i(X_i)$;
\item if $v_i(Y_i ) > v_i(X_i)$, then $\exists\ g\in Y_i$ with $x_{ig}<1$ such that $v_i(Y_i) - v_i(g) < v_i(X_i)$.
\end{enumerate}
\end{corollary}

In other words, \Cref{cor:UGupToOneGood} ensures that, in any deterministic allocation in the $\UG$-decomposition, the valuation of any agent $i$ differs from her expected value by at most the value of one good. Moreover, such a good must have a positive probability of occurring in $i$'s bundle.

\section{Additive Valuations with Entitlements} 
\label{sec:additive}
In this section, we present a lottery for additive valuations that simultaneously achieves ex-ante $\WSD$-$\EF$ (and hence ex-ante $\WEF$) and ex-post $\WEF(1,1)+\WPROP1$. In contrast to equal entitlements, we show a weaker ex-post guarantee. However, we prove that this is necessary since no stronger envy notion is compatible with ex-ante $\WEF$. We also generalize a result of Freeman et al.~\cite{freemanSV20} to entitlements: Similarly to the unweighted setting, we design a lottery that is ex-ante $\WGF$ and ex-post $\WEF_1^1 + \WPROP1$ .

\subsection{Ex-ante $\WSD$-$\EF$, Ex-post $\WEF(1,1)$ + $\WPROP1$}\label{sec:additiveEF} 

The main contribution of this subsection is as follows.

\begin{theorem}\label{thm:mainAdditive}
For entitlements and additive valuations, we can compute in strongly polynomial time a lottery that is ex-ante $\WSD$-$\EF$ and ex-post $\WPROP1$ + $\WEF(1,1)$. 
\end{theorem}

Let us start by introducing our main algorithm \textsc{DifferentSpeedsEating} (\textsc{DSE}), which is inspired by \textsc{Eating} for equal entitlements in~\cite{aziz20}. Agents continuously eat their most preferred available good at speed equal to her entitlement. Every agent starts eating her most preferred good; as soon as a good has been completely eaten it is removed from the set of available goods. Each agent that was eating this good continues eating her most preferred remaining one. The procedure terminates when no good remains. See \Cref{algo:DSEshort} for a formal description. Observe that by precomputing the times at which goods are removed, we can implement the algorithm in strongly polynomial time.

\begin{algorithm}[tbh!]
	\SetNoFillComment
	\DontPrintSemicolon
	\KwIn{An instance $\instance = (\agents, \goods, \set{v_i}_{i\in\agents})$ and the entitlements $w_1, \dots, w_n$}
	\KwOut{A fractional allocation $X$}
	$X\gets \mathbf{0}_{n\times m}$ \tcp*{current fractional allocation}
	$\mathbf{z}\gets \mathbf{1}_m $ \tcp*{remaining supply of each good}
	
	\While{$\goods\neq\emptyset$}{
		$\mathbf{s}\gets \mathbf{0}_m $ \tcp*{eating speed on each item}
		\For{$i\in\agents$}{
			$g^i\gets \argmax_{g\in\goods}v_i(g)$ \tcp*{most favored item}
			$\mathbf{s}(g^i) \gets \mathbf{s}(g^i) + w_i$ \tcp*{sum speeds of each agent}
		}
		$\mathbf{t}\gets \mathbf{1}_m $ \tcp*{eating time of each item}
		\For{$g\in\goods$}{
			$\mathbf{t}(g)\gets \frac{\mathbf{z}(g)}{\mathbf{s}(g)}$ \tcp*{compute finishing times}
		}
		$t\gets \min_{g\in\goods}\mathbf{t}(g)$ \tcp*{time when first item is finished}
		\For{$i\in\agents$}{
			$x\gets t\cdot w_i$ \tcp*{amount of items eaten by $i$}
			$x_{ig^i} \gets x_{ig^i}+x $ \tcp*{eat fraction of $g^i$}
			$\mathbf{z}(g^i) \gets \mathbf{z}(g^i)  -x$ \tcp*{reduce supply of $g^i$}
		}
		$\goods \gets \goods \setminus \{ g \in \goods \mid \mathbf{t}(g) \le \mathbf{t}(g') \text{ for all } g' \in \goods\}$ \tcp*{remove finished items}%
	}
	\KwRet{X}\\
	\caption{\textsc{DifferentSpeedsEating}\label{algo:DSEshort}}
\end{algorithm}

We denote by $\XDSE$ the output of \textsc{DSE}. The key properties are summarized in the following lemma.

\begin{lemma}
	Let $\XDSE$ be the output of \textsc{DSE}, then
	\begin{enumerate}
	\item $\sum_{g\in \goods} x^{\DSE}_{ig} = w_i\cdot m$ for each $i\in \agents$;
	\item the time needed for agent $i$ to eat one unit of goods is $\frac{1}{w_i}$;
	\item overall, one unit of goods is consumed in one unit of time and, therefore, \DSE\ runs for $m$ time units.
	\end{enumerate}
\end{lemma}

We define the \emph{eating time} of a good $g$ as the point in time when it has been entirely consumed (during a run of \DSE). Whenever an agent starts eating a good $g$, she can start eating another good only after the eating time of $g$.

Before proceeding to the proof of~\Cref{thm:mainAdditive}, let us provide an example
describing a run of \DSE\ on the instance $\instance^*$ introduced above in~\Cref{example:valuations}.

\begin{example}[$\DSE$ at work]\label{example:DSE}
	Let us observe the behavior of $\DSE$ on $\instance^*$.
	
	The agents' priorities for the goods are the following:
	\begin{align*}
		g_1 \succ_1 g_2 \succ_1 g_3 \succ_1 g_4 \ ,\\
		g_2 \succ_2 g_3 \succ_2 g_1 \succ_2 g_4 \ , \\
		g_2 \succ_3 g_3 \succ_3 g_1 \succ_3 g_4 \ .
	\end{align*}
	Notice that, for agent $1$, goods $g_1$ and $g_2$ are identical and ties are broken in favor of the good coming first in the ordering $g_1, \dots, g_4$. Moreover, agents $2$ and $3$ have the same priority order, for this reason they will always be eating the same good.
	
	During a run of \DSE, Whenever a good gets entirely eaten up, the behavior of agents who were eating this good changes. In the following, we only refer to time points where these events happen. Indeed, the times are the eating times of the good(s) that have been completely consumed.
	
	\medskip
	
	\noindent {\em\bf Time $\mathbf{t=0}$ :} At the beginning, $x_{ig}=0$, for all $i\in \agents$ and $g\in \goods$. Agent $1$ starts eating $g_1$ while agents $2$ and $3$ good $g_2$. Notice that agents $2$ and $3$ together have the same speed as agent $1$.
	
	\medskip 
	
	\noindent{\em\bf Time $\mathbf{t=2}$ :} $g_1$ and $g_2$ get fully consumed and $x_{1 g_1}=1$, $x_{2g_2}= \frac{2}{3}$ and $x_{3g_2}=\frac{1}{3}$, respectively. Agent $1$ will start eating $g_3$ as well as agents $2$ and $3$. All the agents together have speed equal to $1$. Notice that agent $1$ would prefer good $g_2$, however, it has been consumed entirely by agents $2$ and $3$.
	
	\medskip
	
	\noindent{\em\bf Time $\mathbf{t=3}$ :} $g_3$ is now fully consumed. We have $x_{1 g_3}= \frac{1}{2}$, $x_{2g_3}=  \frac{1}{3}$ and $x_{3g_3}=\frac{1}{6}$, respectively. The only remaining available good is $g_4$, all the agents are now starting to eat it.
	
	\medskip
	
	\noindent{\em\bf Time $\mathbf{t= 4}$ :} All goods are fully consumed and $x_{1 g_4}= \frac{1}{2}$, $x_{2g_4}=  \frac{1}{3}$ and $x_{3g_4}=\frac{1}{6}$. \DSE\ returns the fractional allocation:
	
	\begin{equation*}
		\XDSE=
		\begin{pmatrix} \smallskip
			1 & 0              & \frac{1}{2} & \frac{1}{2} \\ \smallskip
			0 & \frac{2}{3} & \frac{1}{3} & \frac{1}{3}			  \\ \smallskip
			0 & \frac{1}{3} & \frac{1}{6}& \frac{1}{6}
		\end{pmatrix}
		\ .
	\end{equation*}
\exampleend
\end{example}

Our first result is that the output of \DSE\ is $\WSD$-$\EF$.

\begin{proposition}\label{prop:exanteWSD}
$\XDSE$ is $\WSD$-$\EF$.
\end{proposition}
\begin{proof}
For convenience, we use $X=\XDSE$. Let us consider an agent $i\in \agents$. Note that the goods $g_1, \dots, g_m$ are ordered in the same manner as in \DSE\ for agent $i$, since we always break ties according to a predefined ordering of $\goods$. Now consider another agent $j\in \agents$. Using the notation $G_k=\set{g_1, \dots, g_k}$ for the first $k$ goods in $i$'s ordering, we show 
\begin{equation}\label{eq:WSDproof}
	w_j\cdot \sum_{g\in G_k }  x_{ig} \geq w_i\cdot \sum_{g\in G_k }  x_{jg} \enspace ,
\end{equation}
for every $k\in [m]$, and $\WSD$-$\EF$ follows for agent $i$.

Let $t_k$ be the time when $i$ stops eating $g_k$ during the run of \DSE. We set $t_k=t_{k-1}$ if good $g_k$ has been completely consumed before time $t_{k-1}$ by others. This means that, by the time $t_k$, no good in $G_k$ remains available.
On the one hand, until time $t_k$, agent $i$ could only consume goods in $G_k$, implying $w_i\cdot t_k = \sum_{g\in G_k } x_{ig}$. On the other hand, every good in $G_k$ has been fully consumed by that time, i.e., $w_j\cdot t_k \geq \sum_{g\in G_k }  x_{jg}$, for every $j\in \agents$. Combining these two properties proves \Cref{eq:WSDproof} and, hence, the theorem.
\end{proof}

It is known that $\SD$-$\EF$ implies ex-ante $\EF$ for additive valuations; it remains true for different entitlements. We refer to the appendix for the formal proof of the following proposition.

\begin{proposition}\label{prop:WSEEFimpliesWEFexante}
	Given a fractional allocation $X$, if $X$ is ex-ante $\WSD$-$\EF$, then $X$ is ex-ante $\WEF$. 
\end{proposition}

\Cref{prop:exanteWSD} and \Cref{prop:WSEEFimpliesWEFexante} show that the outcome of \DSE\ satisfies the ex-ante properties stated in \Cref{thm:mainAdditive}.

So far we have shown ex-ante properties of lotteries having the output of the \DSE\ as fractional matrix representation.
In the remaining part of this subsection, we show how to get good properties ex-post. In particular, we start from the output of the \DSE, namely, $\XDSE$. We apply the Budish's decomposition with the bihierarchy $\UG$. 

Before proceeding, we give an example of such a decomposition as well as some insights on the guarantees obtained thanks to the  $\UG$ bihierarchy. We again make use of instance $\instance^*$; the allocation $\XDSE$ was computed in~\Cref{example:DSE}.

\begin{example}[The $\mathcal{H}^{\mathsf{UG}}$-decomposition]\label{example:decomposition}
	The $\UG$-decomposition of $\XDSE$ is a convex combination $\lambda_1 Y^1 + \dots +\lambda_k Y^k $, for some integer $k$. Every allocation $Y^h$ is deterministic and its properties are determined by the bihierarchy $\UG$. In the following, we use $Y$ to refer to a generic deterministic allocation in the decomposition.
	
	Recall that $\UG= \mathcal{H}_1 \cup \mathcal{H}_2$. Interpreting an allocation as a matrix, the hierarchy $\mathcal{H}_1$ represents columns and only ensures that any $Y$ is complete.
	
	Let us now consider $\mathcal{H}_2$. Recall, $\mathcal{H}_2= \set{(i, S)\ \vert \ i\in \agents, S\in \mathcal{S}_i} \cup \set{(i,g)\vert i\in \agents, g\in \goods}$.
	
	By~\cref{thm:budishDecomp}, every matrix of an allocation $Y$ in the $\UG$-decomposition is feasible under the quotas $\underline{q}_A= \floor{x^\DSE_A}$ and $\overline{q}_A= \ceil{x^\DSE_A}$, for every $A\in \mathcal{H}_2$, where $x^\DSE_A= \sum_{(i,g)\in A} x^\DSE_{ig}$.
	Note that only one agent appears in any pair of $\mathcal{H}_2$. Hence, we discuss the implications of~\cref{thm:budishDecomp} agent by agent.
	
	\paragraph{Agent $1$:} The pair $(1, S)$ belongs to $\mathcal{H}_2$ if and only if  $S\in \mathcal{S}_1\cup\set{ \set{g_2}, \set{g_3}, \set{g_4}}$, where $\mathcal{S}_1 =\set{\set{g_1}, \set{g_1, g_2}, \set{g_1, g_2, g_3}, \set{g_1,g_2, g_3,g_4}}$.
	The feasibility conditions imply:
	
	\begin{align*}
		&y_{1g_1} =1,  \\
		&y_{1g_1} + y_{1g_2} =1 ,\\
		&1 \leq y_{1g_1} + y_{1g_2} + y_{1g_3}  \leq 2 , \\
		& y_{1g_1} + y_{1g_2} + y_{1g_3}+ y_{1g_4}  = 2 \ ,
	\end{align*}
	and
	\begin{align*}
		y_{1g_2} =0, && 0\leq  y_{1g_3} \leq 1, && 0\leq y_{1g_4} \leq 1 \ .
	\end{align*}
	
	In other words, in any deterministic allocation $Y$, agent $1$ always receives $2$ goods. In particular, she always gets $g_1$ but never $g_2$. Moreover, she gets either $g_3$ or $g_4$, but not both of them.
	
	\paragraph{Agent $2$:} The pair $(2, S)$ belongs to $\mathcal{H}_2$ if and only if  $S\in \mathcal{S}_2\cup\set{\set{g_1}, \set{g_3}, \set{g_4}}$, where $\mathcal{S}_2=\set{\set{g_2}, \set{g_2, g_3}, \set{g_2, g_3, g_1}, \set{g_2, g_3, g_1,g_4}}$.
	In this case, the feasibility conditions imply
	\begin{align*}
		&0\leq y_{2g_2} \leq 1,  \\
		&   y_{2g_2} + y_{2g_3} = 1 ,\\
		&   y_{2g_2} + y_{2g_3} + y_{2g_1}  = 1 , \\
		& 1 \leq  y_{2g_2} + y_{2g_3} + y_{2g_1} + y_{2g_4}  \leq  2 \ ,
	\end{align*}
	and
	\begin{align*}
		y_{2g_1} =0, && 0\leq  y_{2g_3} \leq 1, && 0\leq  y_{2g_4} \leq 1 \ .
	\end{align*}
	
	Therefore, the bundle of agent $2$ is of size either $1$ or $2$. It never contains $g_1$, but must contain one good between $g_2$ and $g_3$, and possibly contains $g_4$.

	\paragraph{Agent $3$:} The pair $(3, S)$ belongs to $\mathcal{H}_2$ if and only if  $S\in \mathcal{S}_3\cup\set{\set{g_1}, \set{g_3}, \set{g_4}}$, where  $\mathcal{S}_3=\set{\set{g_2}, \set{g_2, g_3}, \set{g_2, g_3, g_1}, \set{g_2, g_3, g_1,g_4}}$.
	In this case, $0 \leq  y_{3g_2} + y_{3g_3} + y_{3g_1} + y_{3g_4}  \leq  1$ and $ y_{3g_1}=0$, hence, agent $3$ can receive at most one of $g_2,g_3, g_4$ and never receives $g_1$. 
	
	Finally, we provide a concrete $\UG$-decomposition of $\XDSE$ for $\instance^*$. Considering that rows represent agents and columns represent goods, it is easy to verify that every deterministic allocation satisfies the aforementioned properties.
	\begin{equation*}
	\XDSE=
	\frac{1}{6} \cdot
	\underbrace{
		\begin{pmatrix}
			1 & 0 & 0 & 1 \\
			0 & 1 & 0 & 0  \\
			0 & 0 & 1 & 0 
		\end{pmatrix}
	}_{ Y^1} 
	\ + \
	\frac{1}{6} \cdot
	\underbrace{	
		\begin{pmatrix}
			1 & 0 & 1 & 0  \\
			0 & 1 & 0 & 0  \\
			0 & 0 & 0 & 1 
		\end{pmatrix}
	}_{ Y^2} 
	\ + \
	\frac{1}{3} \cdot
	\underbrace{
		\begin{pmatrix}
			1 & 0 & 0 & 1  \\
			0 & 0 & 1 & 0  \\
			0 & 1 & 0 & 0 
		\end{pmatrix}
	}_{ Y^3} 
	\ + \
	\frac{1}{3} \cdot
	\underbrace{
		\begin{pmatrix}
			1 & 0 & 1 & 0  \\
			0 & 1 & 0 & 1  \\
			0 & 0 & 0 & 0 
		\end{pmatrix}
	}_{ Y^4} 
	\ .
\end{equation*}
\exampleend
\end{example}
\bigskip

We notice that $Y^4$ in~\cref{example:decomposition} is not $\WEF1$. Indeed, in $Y^4$ agents $1$ and $2$ receive two goods each while agent $3$ has an empty bundle, thus agent $3$ $\WEF1$-envies any other agent. On the other hand, every allocation $Y^h$, for $h=1, \dots, 4$, is $\WEF(1,1)$ and $\WPROP1$. We next show this is always the case for the $\UG$-decomposition of any $\XDSE$.

\begin{theorem}\label{thm:expostWEF11}
 Every deterministic allocation $Y$ in the $\UG$-decomposition of $\XDSE$ is $\WEF(1,1)$. 
\end{theorem}

To show the theorem we need some preliminary notions.

\paragraph{Goods Eaten by $i$ at Time $t$.}
Recall that \DSE\ runs for $m$ units of time. Every agent $i$ exactly eats a total mass of $w_i\cdot m$ of $\goods$ during \DSE. Let $g_1, \dots, g_m$ be the ordering of goods according to $v_i$. We define $\Eaten(i, t)= \set{g_1, \dots, g_\ell} = G_\ell$, where $g_{\ell}$ is either a good that agent $i$ just finished to consume (i.e., $t$ is the eating time of $g_{\ell}$ and agent $i$ was consuming it) or agent $i$ at time $t$ is eating the good $g_{\ell+1}$, which has not been finished yet. Consequently, by time $t$, agent $i$ may have contributed only to the consumption of goods in $G_\ell$. In particular, all goods in $G_\ell$ have been entirely consumed (by $i$ or others), since otherwise $i$ would not start eating $g_{\ell+1}$. 

Recall that $w_i$ is the speed of $i$. At time $t = \frac{k}{w_i}$ agent $i$ ate a total mass $k$ of goods. With the next lemma, we show that the $\UG$-decomposition guarantees agent $i$ deterministically receives at most $k$ goods from the ones eaten by time $\frac{k}{w_i}$.

\begin{lemma}\label{lemma:decompositionImplication}
Given any deterministic allocation $Y$ in the $\UG$-decomposition of $\XDSE$, for every $i\in \agents$ and $k=1, \dots, \floor{w_i\cdot m}$, $ \modulus{Y_i \cap \Eaten(i, \frac{k}{w_i})}\leq k$. Furthermore, $\floor{w_i\cdot m}\leq \modulus{Y_i}\leq \ceil{w_i\cdot m}$.
\end{lemma}
\begin{proof}
By definition, $\Eaten(i, \frac{k}{w_i})= G_\ell$, the $\ell$ most preferred goods of $i$, for some $\ell$. Thus, by the time $\frac{k}{w_i}$, agent $i$ only ate goods in $G_\ell$ and possibly is currently eating the next less preferred good.  Moreover,  goods are eaten by $i$ in the same ordering we used to build the collection $\mathcal{S}_i$ in the definition of $\UG$ implying $(i, G_\ell) \in \UG$. Since $\modulus{Y_i \cap \Eaten(i, \frac{k}{w_i})} = \sum_{g\in G_\ell} y_{ig} $, the $\UG$-decomposition properties imply $\sum_{g\in G_\ell} y_{ig} \leq \big\lceil\sum_{g\in G_\ell} x_{ig}\big\rceil $.
This last is upper-bounded by $k$ because of these two simple observations: $g_\ell$ is fully consumed by the time $ \frac{k}{w_i}$, and  at that time agent $i$ ate $k$ units of goods. The first claim follows.

The second claim immediately follows by the $\UG$-decomposition properties, since $(i, \goods)\in \UG$. 
\end{proof}

Given any deterministic allocation $Y$ in the $\UG$-decomposition, consider agent $i$ and sort the goods in $Y_i$ in a non-increasing manner with respect to $v_i$: $Y_i=\set{g^i_{1}, \dots, g^i_{h_i}}$ and $v_i(g^i_{1}) \geq  \dots \geq  v_i(g^i_{h_i})$.
By \Cref{lemma:decompositionImplication}, we see $h_i= \floor{w_i\cdot m}$ or $h_i= \ceil{w_i\cdot m}$. 

\paragraph{Stopping vs.\ Eating Time.}
Given any deterministic allocation $Y$ in the $\UG$-decomposition of $X$, for each $i\in \agents$ and $k\in [h_i]$, we define the stopping time by $s(g_k^i)= \min\set{t(g_k^i), \frac{k}{w_i}}$. Here $t(g_k^i)$ is the time when $g_k^i$ has been entirely consumed during the \DSE, i.e., the eating time of $g_k^i$. Note that $s(g_k^i)$, differently from $t(g_k^i)$, depends on $Y$: Indeed, in $Y_i$ good $g_k^i$ is the $k$-th most preferred good. However, if the eating time is greater than $\frac{k}{w_i}$, this good might appear as $(k+1)$-th most preferred good in another deterministic allocation of the decomposition. For convenience, we omit $Y$ in the notation since we only discuss stopping times of single allocations. Let us show a couple of useful properties of stopping times.
\begin{lemma}\label{lemma:stoppingTimeBounds}
Given any deterministic allocation $Y$ in the $\UG$-decomposition of $\XDSE$, let $g_k^i$ be the $k$-th most preferred good in $Y_i$, it holds $s(g_k^i)\in \left(\frac{k-1}{w_i},\frac{k}{w_i}\right]$.
\end{lemma}
\begin{proof}
	By definition, $s(g_k^i)= \min\set{t(g_k^i), \frac{k}{w_i}} \leq \frac{k}{w_i}$. 
	For contradiction, suppose $t(g_k^i) \leq  \frac{k-1}{w_i}$. Then, $g_k^i \in Y_i \cap \Eaten\left(i, \frac{k-1}{w_i}\right)$. Notice that $t(g_1^i)\leq \dots \leq t(g_k^i)$, by definition of \DSE, and therefore $g_h^i \in Y_i \cap \Eaten\left(i, \frac{k-1}{w_i}\right)$, for each $h=1,\dots, k$. In conclusion, $\left| Y_i \cap \Eaten\left(i, \frac{k-1}{w_i}\right)\right|\geq k$ which is a contradiction by \Cref{lemma:decompositionImplication}, and hence $t(g_k^i) >  \frac{k-1}{w_i}$.
\end{proof}
For the eating time $t(g_k^i)$ the same lower bound holds, but we can only upper bound it by $\frac{k+1}{w_i}$. This difference will be crucial in the proof of~\cref{thm:expostWEF11} and requires the definition of stopping times.

\begin{lemma}\label{lemma:stoppingTimeRelations}
Given any deterministic allocation $Y$ in the $\UG$-decomposition of $\XDSE$, let $g_k^i$ be the $k$-th most preferred good in $Y_i$. For every good $g$ coming earlier in $i$'s ordering of goods, it holds that $s(g) < s(g_k^i)$.
\end{lemma}
\begin{proof}
	The claim follows by the definition of stopping time and the properties of \DSE.
	Indeed, by the definition of stopping time $s(g)\leq t(g)$, and $ t(g) < \min\set{t(g_k^i), \frac{k}{w_i}}=  s(g_k^i)$. 
	The second inequality holds because at time $s(g_k^i)$ agent $i$ is eating or finishes to eat $g_k^i$, and $g$ must have been eaten before $i$ starts eating $g_k^i$. Further, the inequality is strict since agent $i$ ate a positive fraction of $g_k^i$ (that is, $x_{ig_k^i} > 0$);  otherwise, since $(i,g_k^i)\in \mathcal{H}_2$, $x_{ig_k^i}= 0$ would imply $y_{ig_k^i}=0$ and, hence, $g_k^i\not\in Y_i$.
\end{proof}

We are now ready to show \Cref{thm:expostWEF11}.

\begin{proof}[Proof of \Cref{thm:expostWEF11}]
Let $Y$ be any deterministic allocation in the $\UG$-decomposition of $\XDSE$.
The proof proceeds as follows: We first generate a picking sequence $\pi$, then show that $Y$ is the output of such a picking sequence, and finally prove that $\pi$ satisfies \Cref{thm:pickingWEFxy}, for $x=y=1$. This shows that $Y$ is $\WEF(1,1)$.

\smallskip
 
\noindent \textit{Defining $\pi$. $\;$} We sort the goods $\goods$ in a non-increasing order of stopping times $s_1, \dots, s_m$ (defined according to $Y$). If $g\in Y_i$ is the $h$-th good in this ordering, then $\pi(h)=i$.

\smallskip

\noindent \textit{$Y$ is the result of $\pi$. $\;$} Assume $i$ is the $h$-th agent in $\pi$. Assume that $\pi(h)=i$ is the $k$-th occurrence of $i$ in $\pi$. We show that for each $h\in[m]$, the most preferred available good for $i$ is exactly $g_k^i$.
Let us proceed by induction on $h$. 

For $h=1$, clearly, $k=1$. By \Cref{lemma:stoppingTimeRelations}, $g_1^i$ must be the most preferred good of $i$, otherwise we contradict the fact that $s_1 = s(g_1^i)$ is the minimum stopping time. At this point no good has been assigned, so $i$ selects $g_1^i$.

Assume the statement is true until the $h$-th component of $\pi$. We show it is true for $h+1 \leq m$. Suppose a good $g$ coming before $g_k^i$, in $i$'s ordering, is still available. By \Cref{lemma:stoppingTimeRelations}, there exists $h'$ s.t. $s_{h'}=s(g)<s(g_k^i)$ with $h'\leq h$. By the inductive hypothesis, $g$ must have been assigned to $\pi(h')$. On the other hand, $g_k^i$ is still available, otherwise there exists $h'\leq h$, such that $\pi(h')$ picked $g_k^i$ during the $h'$-th round -- a contradiction with the inductive hypothesis.

\smallskip

\noindent \textit{$\pi$ satisfies \Cref{thm:pickingWEFxy}. $\;$}
We now show that $\pi$ satisfies $\WEF(1,1)$. Consider any prefix of $\pi$ and any pair of agents $i,j$. Let us denote by $t_i$ (resp.\ $t_j$) the number of picks of agent $i$ (resp.\ $j$) in the considered prefix.  
Let $s_j$ and $s_i$ be the stopping times of the good selected by $j$ at her $t_j$-th pick and the stopping time of the good selected by $i$ at her $(t_i+1)$-th pick, respectively. If $i$ has no $(t_i+1)$-th pick, we set $s_i = m < \frac{t_i+1}{w_i}$. Within the considered  prefix of $\pi$, agent $j$ already made its $t_j$-th pick but $i$ didn't make its $(t_i+1)$-th pick. Now by definition of $\pi$, $s_j\leq s_i$. By~\Cref{lemma:stoppingTimeBounds}, $s_j >\frac{t_j -1}{w_j} $ and $s_i \leq \frac{t_i+1}{w_i}$. We finally get $ \frac{t_j -1}{w_j}< \frac{t_i+1}{w_i}$. This shows that the hypothesis of \Cref{thm:pickingWEFxy} is fulfilled for $x=y=1$.
%
\end{proof}

Note that if we had chosen eating rather than stopping times for the picking sequence, we could only deduce $\frac{t_j -1}{w_j}< \frac{t_i+2}{w_i}$ which is not sufficient to show $\WEF(1,1)$. 

As $\XDSE$ is (ex-ante) $\WEF$, it is also $\WPROP$. By ex-ante $\WPROP$ and \Cref{cor:UGupToOneGood}, the following holds.
\begin{proposition}
	\label{prop:expostWPROP1}
	Every deterministic allocation $Y$ in the $\UG$-decomposition of $\XDSE$ is $\WPROP1$. 
\end{proposition}
\begin{proof}
The fractional allocation $\XDSE$ is $\WEF$, and hence $\WPROP$. Therefore, $v_i(X_i)\geq w_i\cdot v_i(\goods)$. By \Cref{cor:UGupToOneGood}, for any $Y$ in the $\UG$-decomposition, $v_i(Y_i)\geq v_i(X_i)- v_i(g)$, for some $g\in \goods\setminus Y_i$. This implies $v_i(Y_i \cup \set{g})\geq w_i\cdot v_i(\goods)$, and $\WPROP1$ follows.
\end{proof}
In conclusion, we proved that the $\UG$-decomposition of $\XDSE$ is a lottery achieving ex-ante $\WSD$-$\EF$, and therefore ex-ante $\WEF$, and ex-post $\WEF(1,1)+\WPROP1$. As a consequence of \Cref{thm:budishDecomp}, our lottery has polynomial support and the computation requires strongly polynomial time.

While our guarantee is weaker than the ex-post $\EF1$ for equal entitlements, we show that our lottery is, in a sense, best possible in terms of ex-post guarantees. Indeed, we prove that no stronger ex-post envy notion is compatible with ex-ante $\WEF$.
\begin{proposition}\label{prop:negativeWEF}
  For every pair $x,y \in [0,1]$ such that $x+y<2$, ex-ante $\WEF$ is incompatible with ex-post $\WEF(x,y)$.
\end{proposition}
\begin{proof}
Consider a fair division instance $\instance=(\agents,\goods, v)$, with $\agents=\set{1,2}$ and $\goods=\set{g_1, g_2}$. Moreover, $v_i(g_1)= v_i(g_2)=1$, for $i=1,2$. Let us set $w_1\in \left(\frac{y}{2+y-x},\frac{1}{2} \right)$ and $w_2=1-w_1$. Observe that  $\frac{y}{2+y-x}<\frac{1}{2}$, since $x+y<2$. In any ex-ante $\WEF$ allocation agent $1$ receives in expectation less than one good. This means, the allocation $Y=(Y_1,Y_2)= (\emptyset, \goods)$  is in the support of any ex-ante $\WEF$ lottery. Therefore, since $w_1+w_2=1$, for each $g\in Y_2$,
\begin{align*}
 w_1\cdot \left(v_1(Y_2) -x\cdot v_1(g)\right) = w_1\cdot (2-x)  > \frac{y}{2+y-x} \cdot (2-x) 
> w_2\cdot  y  = w_2\cdot \left(v_1(Y_1) + y\cdot v_1(g)\right) \ .
\end{align*}
This proves $Y$ is not $\WEF(x,y)$.
\end{proof}

\paragraph{Remark: Equal Entitlements Case.}
Let us remark that for equal entitlements our approach also provides ex-ante $\EF$ and ex-post $\EF1$. The ex-ante property follows directly since $w_i=1/n$. For ex-post $\EF1$, similarly to~\cite{aziz20}, it is possible to show that any allocation $Y$ in the $\UG$-decomposition of the $\XDSE$ is the result of an RB picking sequence. In particular, this holds for the picking sequence defined in the proof of \cref{thm:expostWEF11}.


\subsection{Ex-ante $\WGF$ and Ex-post $\WPROP1$ +  $\WEF_1^1$}\label{sec:additiveGroupFair}

In this subsection, we generalize a result of Freeman et al.~\cite{freemanSV20} to entitlements. We follow the general argument and incorporate some technical extensions to allow for different agent weights.

For our purposes, we recall the weighted version of the well known Nash Welfare.
\begin{definition}[Weighted Nash Welfare]
	Given a fair division instance $\instance$, with entitlements $w_1, \dots, w_n$, and  an allocation $\allocation=(A_1, \dots, A_n)$,  the {\em weighted Nash welfare} of $\allocation$ is given by
	$
	\prod_{i=1}^n \left(v_i(A_i)\right)^{w_i} \ .
	$
\end{definition}

\begin{theorem}\label{thm:groupFair}
	For entitlements and additive valuations, we can compute in strongly polynomial time a lottery that is ex-ante $\WGF$ and ex-post $\WPROP1$ + $\WEF_1^1$.
\end{theorem}
\begin{proof}
	We follow the proof steps of~\cite{freemanSV20}, where the same result was shown in the unweighted setting. In a first step, we show that fractional maximum weighted Nash welfare allocations (MWN allocations) form competitive equilibria (CE) (\Cref{lemma:MWN->CE}), and in a second step, we show that CE allocations are $\WGF$ (\Cref{lemma:CE->WGF}). Finally, we show that the $\UG$-decomposition of any fractional MWN allocation yields the desired ex-post properties (\Cref{lemma:expostByUtilityGuarantee}).
	
	The next lemmas hold even in the more general cake cutting model introduced by Steinhaus~\cite{steinhaus}. In this model, instead of a finite number of goods, a single continuous good $C$ (a ``cake'') needs to be split among the $n$ agents. An allocation is a partition of $C$ into $n$ subsets, and the valuation of an agent $i\in\agents$ is given by a measure $v_i$ on $C$.
	
	First, we use a slight generalization of~\cite[Lemma 4.8]{haleviEtAl} to derive a useful inequality for MWN allocations. 
	\begin{lemma}\label{lemma:stromquist}
		Let $f_i : \RR^+ \mapsto \RR^+$, $i\in \agents$, be differentiable functions. If an allocation $X$ of cake maximizes the welfare function $S(X)=\sum_{i\in \agents} f_i(v_i(X_i))$, then for any two agents $i,j\in \agents$ and any slice $Z_j\subseteq X_j$,
		\[
		f_j'(v_j(X_j)) \cdot v_j(Z_j)\geq f_i'(v_i(X_i)) \cdot v_i(Z_j).
		\]
	\end{lemma}
	\begin{proof}
		Like in \cite{haleviEtAl}, we use a result of \cite{stromquist}, which states that for any part $Z$ of the cake and $\alpha \in [0,1]$, there is $Z_\alpha \subseteq Z_j$ such that $v_i(Z_\alpha) = \alpha \cdot v_i(Z)$ and $v_j(Z_\alpha) = \alpha \cdot v_j(Z)$. Let $X'$ be the allocation obtained by giving $Z_\alpha$ from agent $j$ to $i$. The welfare difference $D = S(X') - S(X)$ is now a function of $\alpha$, i.e.,
		\begin{align*}
			D(\alpha) 
			=\ & f_i\big(v_i(X_i')\big) - f_i\big(v_i(X_i)\big) 
			+ f_j\big(v_j(X_j')\big) - f_j\big(v_j(X_j)\big) \\
			= \ & f_i\big(v_i(X_i) + \alpha\cdot v_i(Z_j)\big) - f_i\big(v_i(X_i)\big) 
			+ f_j\big(v_j(X_j) - \alpha\cdot v_j(Z_j)\big) - f_j\big(v_j(X_j)\big).
		\end{align*}
		The derivative of $D$ is 
		\begin{align*}
			D'(\alpha) = \ & v_i(Z_j) \cdot f_i'\big(v_i(X_i)
			+ \alpha\cdot v_i(Z_j)\big) - v_j(Z_j) \cdot f_j'\big(v_j(X_j) - \alpha\cdot v_j(Z_j)\big).
		\end{align*}
		Since $X$ maximizes the welfare $S(X)$, we must have $D'(0)\leq 0$, and hence
		$
		v_i(Z_j) \cdot f_i'(v_i(X_i) ) \leq v_j(Z_j) \cdot f_j'(v_j(X_j) ).
		$
	\end{proof}
	
	If we set $f_i = w_i \ln v_i(X_i)$ in \Cref{lemma:stromquist}, we obtain the following useful inequality. 
	\begin{corollary}\label{cor:usefulIneq}
		If $X$ is an MWN allocation, then for any two agents $i,j\in \agents$ and $Z_j\subseteq X_j$, 
		\[
		w_j \cdot \frac{v_j(Z_j)}{v_j(X_j)} \geq w_i \cdot \frac{v_i(Z_j)}{v_i(X_i)} \enspace.
		\]
	\end{corollary}
	
	For the next proofs we need the notion of competitive equilibrium.
	\begin{definition}
		A pair $(X,P)$ of allocation and prices is a competitive equilibrium (CE), if
		\begin{enumerate}[label=\arabic*)]
			\item $P(Z)>0$ iff $Z$ is a positive slice\footnote{A part $Z$ of the cake is called positive slice if $v_i(Z)>0$ for at least one agent $i\in \agents$.}, 
			\item for all agents $i\in \agents$, $Z_i\subseteq X_i$, and slice $Z$, $\frac{v_i(Z_i)}{P(Z_i)} \geq \frac{v_i(Z)}{P(Z)}$ (MBB),
			\item for all $i\in \agents$, $P(X_i) = w_i$.
		\end{enumerate}
	\end{definition}
	This equilibrium notion stems from \cite{haleviEtAl}, where the authors define it only for equally entitled agents and call it \emph{strong competitive equilibrium from equal incomes} (sCEEI).
	
	\begin{lemma}\label{lemma:MWN->CE}
		For every MWN allocation $X$ there exists a price measure $P$ such that $(X,P)$ is a competitive equilibrium.
	\end{lemma}
	\begin{proof}
		Let $X$ be an MWN allocation. For an agent $i\in\agents$, we define the price of a slice $Z_i\subseteq X_i$ as
		\[
		P(Z_i) = w_i\cdot \frac{v_i(Z_i)}{v_i(X_i)}.
		\]
		The price of an arbitrary slice $Z$ is then given by adding the agent-specific parts, i.e.,
		\[
		P(Z) = \sum_{i\in \agents} P(Z\cap X_i).
		\]
		We need to show that conditions 1) to 3) are fulfilled by $(X,P)$.
		Condition 3) follows simply from the definition of prices. Condition 1) can be derived as follows: Suppose $P(Z)>0$ for some slice $Z$. Then, again by the definition of prices, there must be at least one agent $i\in \agents$ such that $P(Z\cap X_i)>0$. Hence $Z$ is a positive slice. For the other direction, suppose $Z$ is a positive slice, i.e., $v_i(Z) > 0$ for some agent $i\in\agents$. Then either $P(Z\cap X_i)>0$ or $P(Z\cap X_i)=0$. In the first case we are done. In the second case there needs to be an agent $j\in\agents$ with $P(Z\cap X_j)>0$ as otherwise $Z$ would be liked by $i$ but given only to agents $j$ not liking it, a contradiction to $X$ being an MWN allocation.
		
		Now it remains to show condition 2). 
		Consider any agent $i$, any $Z_i\subseteq X_i$, and any slice $Z$. We partition $Z$ into agent specific parts $Z_j = Z\cap X_j$ (note that $Z=\bigcup_j Z_j$).
		By plugging in prices into \Cref{cor:usefulIneq} we obtain for each part $Z_j$, 
		\[
		v_i(X_i) P(Z_j) \geq v_i(Z_j) w_i.
		\]
		Summing over all agent-specific parts $Z_j$ and using additivity of prices yields
		\[
		v_i(X_i) P(Z) \geq v_i(Z) w_i.
		\]
		Finally, by definition of prices, $v_i(X_i) = w_i \frac{v_i(Z_i)}{P(Z_i)}$, and therefore
		$
		\frac{v_i(Z_i)}{P(Z_i)} \geq \frac{v_i(Z)}{P(Z)}.
		$
	\end{proof}
	
	\begin{lemma}\label{lemma:CE->WGF}
		For every competitive equilibrium $(X,P)$, the allocation $X$ is $\WGF$.
	\end{lemma}
	\begin{proof}
		Let $(X,P)$ be a CE. Assume for sake of contradiction that for some $S,T \subseteq \agents$ there is a reallocation $X'$ of $\cup_{j\in T} X_j$ among agents in $S$, such that for all $i\in S$,
		\begin{equation}\label{eq:groupFair}
			\frac{v_i(X'_i)}{v_i(X_i)} \geq \frac{w_T}{w_S},
		\end{equation}
		and at least one inequality is strict. 
		Consider an agent $i\in S$.
		By choosing $Z_i=X_i$ and $Z=X'_i$ in the second condition of CE we obtain
		\[
		P(X'_i) \geq \frac{v_i(X'_i)}{v_i(X_i)}P(X_i) = \frac{v_i(X'_i)}{v_i(X_i)} w_i,
		\]
		where the last equation comes from the third condition of CE.
		Now summing over $i\in S$ and using \eqref{eq:groupFair} yields
		\[
		\sum_{i\in S} P(X'_i) > \frac{w_T}{w_S} \sum_{i\in S} w_i = w_T,
		\]
		which is a contradiction since 
		\begin{align*}
			\sum_{i\in S} P(X'_i) & = P(\cup_{i\in S} X'_i) = P(\cup_{j\in T} X_j) 
			= \sum_{j\in T} P(X_j) = \sum_{j\in T} w_j = w_T.
		\end{align*}
	\end{proof}
	
	It remains to show the ex-post guarantee. At this point, we drop the consideration of the cake cutting setting and shift to the indivisible domain.
	\begin{lemma}\label{lemma:expostByUtilityGuarantee}
		The $\UG$-decomposition of a fractional MWN allocation is $\WPROP1$ and $\WEF_1^1$.
	\end{lemma}
	\begin{proof}
		We proceed like in \cite{freemanSV20}. Let $X$ be a fractional MWN allocation and let $Y$ be any deterministic allocation in the $\UG$-decomposition of $X$.
		
		We first show that $Y$ is $\WPROP1$. Consider an agent $i\in\agents$. Observe that $X$ as an MWN allocation is $\WGF$ and hence $\WPROP$, i.e., it holds $v_i(X_i)\geq w_i\cdot v_i(\goods)$. Hence if $v_i(Y_i)\geq v_i(X_i)$, we are done. Otherwise, it holds $v_i(Y_i) < v_i(X_i)$. From \Cref{cor:UGupToOneGood} we know that in this case there exists $g^*\notin Y_i$ such that $v_i(Y_i) + v_i(g^*) > v_i(X_i)\geq w_i\cdot v_i(\goods)$. Hence $Y$ is $\WPROP1$.
		
		It remains to show that $Y$ is $\WEF_1^1$. Consider two agents $i,j\in\agents$. From \Cref{cor:UGupToOneGood} it follows that if $v_i(Y_i) < v_i(X_i)$, then there exists $g_i\notin Y_i$ such that
		\begin{equation}\label{ineq:agenti}
			v_i(Y_i) + v_i(g_i) > v_i(X_i).
		\end{equation}
		Analogously, if $v_j(Y_j) > v_j(X_j)$, then there exists $g_j\in Y_j$ such that
		\begin{equation}\label{ineq:agentj}
			v_j(Y_j) - v_j(g_j) < v_j(X_j).
		\end{equation}
		Next, we use \Cref{cor:usefulIneq} again, this time in the setting with discrete goods. It implies for any two agents $i,j\in\agents$ and any $g\in \goods$ with $x_{jg}>0$, 
		\begin{equation*}
			w_j\frac{v_j(g)}{v_j(X_j)} \geq w_i\frac{v_i(g)}{v_i(X_i)}.
		\end{equation*}
		Summing over all $g\in Y_j\setminus \{g_j\}$ yields
		\begin{equation*}
			w_j \frac{v_j(Y_j\setminus \{g_j\})}{v_j(X_j)} \geq w_i \frac{v_i(Y_j\setminus \{g_j\})}{v_i(X_i)}.
		\end{equation*}
		Note that the left-hand side of this inequality is strictly less than $w_j$ due to \eqref{ineq:agenti}. Hence, by inequality \eqref{ineq:agentj} it follows $w_j\cdot (v_i(Y_i) + v_i(g_i)) > w_j\cdot v_i(X_i) > w_i\cdot v_i(Y_j\setminus \{g_j\})$. This shows $Y$ is $\WEF_1^1$.
	\end{proof}
	
	\Cref{thm:groupFair} now follows from Lemmas \ref{lemma:MWN->CE}, \ref{lemma:CE->WGF}, \ref{lemma:expostByUtilityGuarantee} and the fact that we can  can compute in strongly polynomial time an MWN allocation \cite{Orlin10} as well as the $\UG$-decomposition (\Cref{thm:budishDecomp}).
\end{proof}

\paragraph{Remark.} One might wonder whether the ex-post guarantee in \Cref{thm:groupFair} could be replaced with $\WEF(x,y)$ for some parameters $x,y\in[0,1]$. There are instances where this is impossible, even in the unweighted setting. Consider the following example: There are three agents $1,2,3$, three light goods, and one heavy good. Agents $1$ and $2$ have the same valuation function, they value the heavy good at 6, and each light good at 1. Agent $3$ values the light goods at 1 and the heavy good at 0.

Now consider a fractional group fair allocation $X$. Observe that all light goods need to be allocated completely to agent $3$ in $X$. If one of the first two agents (say, agent 1) gets a fraction $\epsilon>0$ of light goods, the group fairness condition is violated for $S=\{2,3\}$ and $T=\{1,3\}$: If one reallocates the fraction $\epsilon$ of light goods from agent $1$ to agent $3$, then agent $3$ strictly improves and the utility of agent $2$ remains unchanged.

Now consider any allocation $Y$ in the support of a lottery implementing $X$. $Y$ needs to give all light goods to agent $3$, so at least one of the first two agents gets no good at all. This agent is then envious to agent $3$, and transferring one good from agent $3$ to agent $1$ cannot remove this envy.


\section{Extensions to General Valuations}\label{sec:general}
In this section, we explore to which extent our techniques apply to more general valuations.
A major challenging problem for non-additive valuations is that the (expected) utility of an agent for a lottery is not uniquely determined by its fractional matrix.
Nonetheless, both ex-ante $\WEF$ and ex-ante $\WSD$-$\EF$ allocations exist for all valuations $\set{v_i}_{i\in \agents}$. In particular, for ex-ante $\WEF$ one can simply assign $\goods$ to agent $i$ with probability $w_i$. For ex-ante $\WSD$-$\EF$ it is sufficient to invoke \DSE\ only using agents' priorities over single goods. Observe that for general valuations it is no longer true that ex-ante $\WSD$-$\EF$ implies ex-ante $\WEF$, not even in the unweighted setting.
We next show that ex-ante $\WEF$ and either ex-post $\WPROP1$ or ex-post $\WEF(1,1)$ might no longer be possible.

\begin{theorem}\label{thm:impossibilityGeneralWeighted}
  For general valuations, ex-ante $\WEF$ is not compatible with $\WPROP1$ or $\WEF(1,1)$.
\end{theorem}
\begin{proof}
Let us consider a fair division instance with two agents and four goods. Suppose agent $1$ has an entitlement of $\frac{2}{3}$, and has value $1$ for any bundle (except the empty bundle, for which she has value $0$). Agent $2$ has entitlement $\frac{1}{3}$ and value $k$ for any bundle of size $k$, for $k=0,\dots, 4$. 
	
Let us denote by $p_k$ the probability that $1$ receives $k$ goods. Clearly, if the allocation is ex-post $\WPROP$ or ex-post $\WEF(1,1)$, then $p_0 = 0$ and $\sum_{k=1}^4 p_k = 1$. Since the allocation is complete, $p_k$ is also the probability that agent $2$ receives $4-k$ goods. Agent $1$ is ex-ante $\WEF$ if and only if 
\begin{align*}
\frac{1}{3}\cdot (p_1 + p_2 +p_3+p_4) &\geq \frac{2}{3}\cdot (p_0+p_1 + p_2 +p_3)
\end{align*}
which implies $1-p_0 \geq 2\cdot (1-p_4)$ and, hence, $2p_4 \geq 1 + p_0$. Thus, $p_4\geq \frac{1}{2}$. Therefore, the allocation where agent $1$ receives every good and $2$ no good occurs with positive probability. However, such an allocation is neither $\WPROP1$ nor $\WEF(1,1)$ (not even $\WEF_1^1$) for agent $2$.
\end{proof}
Notice that both agents in the proof value $1$ each good, and, hence, agent $1$ is unit-demand and agent $2$ is additive. This means that as soon as one agent is not additive in the weighted case our positive result no longer holds. We further observe that the proof relies on the fact that agents are asymmetric.
Therefore, we next consider two questions: 1) {\em what combination of ex-ante and ex-post properties we can guarantee in the weighted setting assuming slightly more general valuations (namely, XOS)?}, and  2) {\em which valuations still guarantee ex-ante $\EF$ and $\EF1$ ex-post for symmetric agents?}.

We try to reply to question 1) and 2) in~\Cref{sec:XOS} and~\ref{sec:generalValSymmetricAgents}, respectively.

\subsection{XOS Valuations}\label{sec:XOS}
For an agent $i$ with XOS valuation, our algorithms only make use of the additive function $f_i$ such that $v_i(\goods)= \sum_{g\in \goods} f_i(g)$. We either assume $f_i$ to be known or have access to an XOS-oracle (using which $f_i$ can be obtained with a single query).  Given a query with a set $A\subseteq \goods$, the XOS-oracle returns a function $f\in\mathcal{F}_i$ that maximizes $f(A)$.

Let $X$ be the fractional allocation with $x_{ig} = w_i$, for each $i\in \agents$ and $g\in \goods$.
\begin{proposition}
	\label{prop:prop}
	$X$ is ex-ante $\WPROP$.
\end{proposition}
\begin{proof}
	For any allocation $Y$, $v_i(Y_i)= \max_{f\in\mathcal{F}_i} f(Y_i)\geq f_i(Y_i)$, since $f_i\in \mathcal{F}_i$. Hence, $v_i(X_i) \geq f_i(X_i) = \sum_{g\in\goods} x_{ig} \cdot f_i(g) = w_i\cdot \sum_{g\in\goods} {f_i(g)} = w_i\cdot v_i(\goods)$.
\end{proof}

In order to apply Theorem~\ref{thm:budishDecomp}, we need to set up an appropriate additive function. For the next result, we assume that agent $i$ has additive valuation $f_i$, for each $i\in \agents$.
\begin{proposition}
	The $\UG$-decomposition of $X$ is ex-post $\WPROP1$.
\end{proposition}
\begin{proof}
	Given any allocation $Y$ of the decomposition, by definition of XOS, \Cref{cor:UGupToOneGood} and \Cref{prop:prop}, we see that  $v_i(Y_i \cup \{g\}) \ge f_i(Y_i \cup \{g\}) = f_i(Y_i) + f_i(g) > f_i(X_i) = w_i \cdot v_i(\goods)$.
\end{proof}

\subsection{Equal Entitlements}\label{sec:generalValSymmetricAgents}

Here we discuss to which extent we can guarantee BoBW results for equally entitled agents and general valuations. In particular, we explore valuation functions for which \DSE\ with the $\UG$-decomposition can be used to guarantee ex-ante $\EF$ and ex-post $\EF1$. Recall that \DSE\ was already introduced as the \textsc{Eating} procedure by Aziz~\cite{aziz20} for equally entitled agents. 

Both \DSE\ and the definition of the $\UG$ bihierarchy only depend on the ranking of each agent for singleton bundles of goods. Therefore, we can determine a random allocation with \DSE\ and compute its $\UG$-decomposition for any class of valuation functions. Since the concept of $\SD$-$\EF$ depends only on the ranking of single goods provided by the agents, the output of \DSE\ is always an ex-ante $\SD$-$\EF$ allocation, regardless of the considered valuation functions.

Unfortunately, for general valuations, it is no longer true that an $\SD$-$\EF$ allocation $X$ is ex-ante $\EF$ not even if $X$ is the output of \DSE. Such an impossibility holds also for agents having unit-demand valuations as the following example shows.

\begin{example}
Let $\instance$ be a fair division instance with three agents and three goods. Assume agents to have identical unit-demand valuations; goods have all value $1$. Having all agents the same priorities over goods, the output of \DSE\ is given by the following:
	\begin{equation*}
	\XDSE=
	\begin{pmatrix} \smallskip
		\frac{1}{3} & \frac{1}{3} & \frac{1}{3} \\ \smallskip
		\frac{1}{3} & \frac{1}{3} & \frac{1}{3} \\ \smallskip
		\frac{1}{3} & \frac{1}{3} & \frac{1}{3}
	\end{pmatrix}
	\ .
\end{equation*}
Being $\XDSE$ the output of \DSE\ it is $\SD$-$\EF$; we next provide a decomposition of $\XDSE$ showing that $\XDSE$ is not necessarily ex-ante $\EF$. Let us consider the following decomposition:
	\begin{equation*}
	\XDSE=
	\frac{1}{3} \cdot
	\underbrace{
		\begin{pmatrix}
			1 & 1 & 0  \\
			0 & 0 & 1  \\
			0 & 0 & 0 
		\end{pmatrix}
	}_{ Y^1} 
	\ + \
	\frac{1}{3} \cdot
	\underbrace{	
		\begin{pmatrix}
			0 & 0 & 1  \\
			1 & 0 & 0  \\
			0 & 1 & 0 
		\end{pmatrix}	
	}_{ Y^2} 
	\ + \
	\frac{1}{3} \cdot
	\underbrace{
		\begin{pmatrix}
			0 & 0 & 0  \\
			0 & 1 & 0  \\
			1 & 0 & 1 
		\end{pmatrix}
	}_{ Y^3} 
	\ .
\end{equation*}
Notice that such a decomposition is not an $\UG$-decomposition of $\XDSE$ (in any $\UG$-decomposition of $\XDSE$, every agent receives deterministically exactly one good). We claim that agent $1$ (corresponding to the first row) is not ex-ante $\EF$. In fact, the expected utility she has for her random bundle is given by 
\[
\frac{1}{3}\cdot\left(v_1(Y_1^1) + v_1(Y_1^2)+ v_1(Y_1^3)\right) = \frac{1}{3}\cdot (1+1+0)= \frac{2}{3}\ ;
\]
while the expected value of agent $1$ for the random bundle of agent $2$ (second row) is
\[
\frac{1}{3}\cdot\left(v_1(Y_2^1) + v_1(Y_2^2)+ v_1(Y_2^3))\right) = \frac{1}{3}\cdot (1+1+1)= 1 \ .
\] 
\exampleend
\end{example}

Although for multi-demand valuations $\SD$-$\EF$ does not imply ex-ante $\EF$, we next prove that the $\UG$-decomposition of $\XDSE$ is indeed ex-ante $\EF$. Moreover, such a decomposition also guarantees ex-post $\EF1$.

\begin{theorem}\label{thm:multi-demand}
	For equal entitlements and multi-demand valuations, the $\UG$-decomposition of $\XDSE$ is ex-ante $\EF$ and ex-post $\EF1$.
\end{theorem}
\begin{proof}
	As already observed, any allocation in the $\UG$-decomposition of $\XDSE$ is result of an RB picking sequence (in case of equal entitlements our approach coincides with the one of \cite{aziz20}).
	Hence, ex-post $\EF1$  follows by noticing that any RB picking sequence is $\EF1$ for multi-demand valuations.
	
	We now prove ex-ante $\EF$ of the $\UG$-decomposition of $\XDSE$.
	For convenience, we denote $\XDSE$ by $X$ and assume $i$ to be $k$-unit-demand. 
	Consider any agent $i\in \agents$. Given any $j\in \agents\setminus{\set{i}}$, we need to show that
	$ {v_i(X_i)}\geq  {v_i(X_j)}$. In what follows, we again sort goods $g_1, \dots, g_m$ in the ordering induced by agent $i$ over $\goods$.
	
	Let us first show that $ {v_i(X_i)}= \sum_{g\in \goods} v_i(g) \cdot \overline{x}_{ig}$, where
	\begin{align*}
		\overline{x}_{ig_\ell}=
		\begin{cases}
			{x}_{ig_\ell} &\text{ if } \sum_{h=1}^\ell x_{ig_h} \leq k\\
			k - \sum_{h=1}^{\ell -1} x_{ig_h} & \text{ if } \sum_{h=1}^\ell x_{ig_h} > k \\
			&\text{ and } \sum_{h=1}^{\ell -1} x_{ig_h} < k\\
			0  & \text{ otherwise.} 
		\end{cases}
	\end{align*}
	
	In other words, only the fraction of goods agent $i$ ate during the first $k$ units of time do count.
	Let us denote by $g_{h^*}$ the less valuable good for which $\overline{x}_{ig_{h^*}}>0$. Notice that at time $t=k$ of \DSE\ agent $i$ is eating or finishes eating good $g_{h^*}$.
	
	Recall that any agent needs $1/w_i=n$ units of time to eat $1$ unit of goods.
	The feasibility conditions of \Cref{thm:budishDecomp} ensure that in any allocation $Y$ in the $\UG$-decomposition, the $h$-most preferred good of agent $i$ in her bundle $Y_i$ is one good she ate in the time interval $(n(h-1), nh]$.  Further, only one good eaten in this interval will be the $h$-most preferred in the bundle $Y_i$. Since the valuation of $i$ is $k$-unit-demand, only the goods she ate by the time $kn$ will contribute to her utility. This implies $ {v_i(X_i)}= \sum_{g\in \goods} v_i(g) \cdot \overline{x}_{ig}$. 
	
	We now provide an upper bound on ${v_i(X_j)}$. Let us denote by $p_h$ the probability that $g_h$ is in the bundle $X_j$ and is one of the $k$ most preferred goods of agent $i$ among the ones in $X_j$. The probability $p_h$ is upper-bounded by the probability of having $g_h$ in the bundle of $j$, and therefore $p_h\leq x_{jg_h}$. Moreover,  $\sum_{h=1}^m p_h=\sum_{h=1}^m\sum_{\ell=1}^k p_h^\ell = \sum_{\ell=1}^k\sum_{h=1}^m p_h^\ell\leq k$, where $p_h^\ell$ is the probability that $g_h$ is the $\ell$-th most preferred good of $i$ in $j$'s bundle.
	Notice also that $\sum_{h=1}^{h^*-1} x_{jg_h}\leq k$ since, because of the \DSE, $\sum_{h=1}^{h^*-1} x_{jg_h}\leq \sum_{h=1}^{h^*-1} x_{ig_h}$, and $\sum_{h=1}^{h^*-1} x_{ig_h}\leq k$ by definition of $h^*$.
	Therefore, the expected utility of $i$ for $j$'s bundle is give by 
	\begin{align*}
		{v_i(X_j)} =
		 \sum_{h=1}^m p_h\cdot v_i(g_h) 
		\leq
		\sum_{h=1}^{h^*-1} x_{jg_h}\cdot v_i(g_h) + \left(k-  \sum_{h=1}^{h^*-1} x_{jg_h}\right)\cdot v_i(g_{h^*})  \ ,
	\end{align*}
where the inequality holds because of the aforementioned properties of $p_h$ and $x_{jg_h}$, and the fact that $g_1, \dots, g_m$ are sorted in a decreasing manner with respect to $i$'s valuations.
	Moreover, by stochastic dominance, 
	\begin{align*}
		\sum_{h=1}^{h^*-1} x_{jg_h}\cdot v_i(g_h) + \left(k-  \sum_{h=1}^{h^*-1} x_{jg_h}\right)\cdot v_i(g_{h^*})  
		&\leq   \sum_{h=1}^{h^*-1} x_{ig_h}\cdot v_i(g_h) + \left(k-  \sum_{h=1}^{h^*-1} x_{ig_h}\right)\cdot v_i(g_{h^*})
		\\
		&=   \sum_{g\in \goods} v_i(g) \cdot \overline{x}_{ig} \leq {v_i(X_i)} \ .
	\end{align*}
In conclusion, $v_i(X_j)\leq v_i(X_i)$	and the theorem follows.
\end{proof}

Notice that the proof of \Cref{thm:multi-demand} takes into account only an agent $i$ having multi-demand valuations for a given $k$, and not the valuations of the others. For this reason the theorem holds for multi-demand agents having different demands $k$. Moreover, when $k\geq m$ the agent is additive.  These simple observations lead to the following. 

\begin{corollary}
	For equal entitlements and any combination of additive and multi-demand valuations, the $\UG$-decomposition of $\XDSE$ is ex-ante $\EF$ and ex-post $\EF1$.
\end{corollary}

Turning to more general cancelable valuations, we can show that RB picking sequences (see Appendix \ref{appendix:cancelable}) still provide an $\EF1$ allocation, and therefore the ex-post guarantee is maintained. Unfortunately, we were not able to prove that the lottery is ex-ante $\EF$ (only ex-ante $\SD$-$\EF$), and this remains an interesting open question.

\section{Conclusions and Future Work}
In this paper, we obtain best of both worlds results for fair division with entitlements. Our results for additive valuations paint a rather complete picture. We present a lottery that can be computed in strongly polynomial time and guarantees ex-ante $\WEF$ and ex-post $\WPROP1$ + $\WEF(1,1)$. This is tight in the sense that any stronger notion of $\WEF(x,y)$ is incompatible with ex-ante $\WEF$. We also present a lottery that is ex-ante $\WGF$ and ex-post $\WPROP1$ + $\WEF_1^1$. Again, ex-ante $\WGF$ is incompatible with stronger ex-post notions.

We also explore how some of our results can be extended to more general valuation functions. These insights represent an interesting first step, but many important open problems remain. As a prominent one, to the best of our knowledge, it is open for which classes of valuations functions ex-ante $\EF$ is always compatible with ex-post $\EF1$ in the unweighted setting. In addition, providing tight guarantees with entitlements and combinations of other fairness concepts (such as, e.g., variants of the Max-Min Share ($\mathsf{MMS}$)) is an interesting direction for future work.

\bibliographystyle{plain}
\bibliography{literature}

\begin{thebibliography}{10}

\bibitem{amanatidisBFV22}
Georgios Amanatidis, Georgios Birmpas, Aris Filos-Ratsikas, and Alexandros~A
  Voudouris.
\newblock Fair division of indivisible goods: A survey.
\newblock {\em arXiv preprint arXiv:2202.07551, to appear IJCAI'22}, 2022.

\bibitem{aziz20}
Haris Aziz.
\newblock Simultaneously achieving ex-ante and ex-post fairness.
\newblock In {\em Proc.\ 16th Conf.\ Web and Internet Econ.\ (WINE)}, pages
  341--355. Springer, 2020.

\bibitem{AzizGMW14}
Haris Aziz, Serge Gaspers, Simon Mackenzie, and Toby Walsh.
\newblock Fair assignment of indivisible objects under ordinal preferences.
\newblock In {\em Proc.\ Conf.\ Auton.\ Agents and Multi-Agent Syst.\ (AAMAS)},
  pages 1305--1312, 2014.

\bibitem{azizMS20}
Haris Aziz, Herv{\'e} Moulin, and Fedor Sandomirskiy.
\newblock A polynomial-time algorithm for computing a pareto optimal and almost
  proportional allocation.
\newblock {\em Oper.\ Res.\ Lett.}, 48(5):573--578, 2020.

\bibitem{babaioffEF21}
Moshe Babaioff, Tomer Ezra, and Uriel Feige.
\newblock Best-of-both-worlds fair-share allocations.
\newblock {\em arXiv preprint arXiv:2102.04909}, 2021.

\bibitem{babaioffEF21share}
Moshe Babaioff, Tomer Ezra, and Uriel Feige.
\newblock Fair-share allocations for agents with arbitrary entitlements.
\newblock In {\em Proc.\ 22nd Conf.\ Econ.\ Comput.\ (EC)}, pages 127--127,
  2021.

\bibitem{BogomolnaiaM01}
Anna Bogomolnaia and Herve Moulin.
\newblock A new solution to the random assignment problem.
\newblock {\em J. Econ.\ Theory}, 100:295--328, 2001.

\bibitem{Budish11}
Eric Budish.
\newblock The combinatorial assignment problem: {A}pproximate competitive
  equilibrium from equal incomes.
\newblock {\em J. Policial Econ.}, 119(6):1061--1103, 2011.

\bibitem{budishCK13}
Eric Budish, Yeon-Koo Che, Fuhito Kojima, and Paul Milgrom.
\newblock Designing random allocation mechanisms: Theory and applications.
\newblock {\em Amer.\ Econ.\ Rev.}, 103(2):585--623, 2013.

\bibitem{caragiannisKK21}
Ioannis Caragiannis, Panagiotis Kanellopoulos, and Maria Kyropoulou.
\newblock On interim envy-free allocation lotteries.
\newblock In {\em Proc.\ 22nd Conf.\ Econ.\ Comput.\ (EC)}, pages 264--284,
  2021.

\bibitem{chakrabortyIS21}
Mithun Chakraborty, Ayumi Igarashi, Warut Suksompong, and Yair Zick.
\newblock Weighted envy-freeness in indivisible item allocation.
\newblock {\em {ACM} Trans.\ Econ.\ Comput.}, 9(3):1--39, 2021.

\bibitem{chakrabortySKS21}
Mithun Chakraborty, Ulrike Schmidt-Kraepelin, and Warut Suksompong.
\newblock Picking sequences and monotonicity in weighted fair division.
\newblock {\em Artif.\ Intell.}, 301:103578, 2021.

\bibitem{chakrabortySHS22}
Mithun Chakraborty, Erel Segal-Halevi, and Warut Suksompong.
\newblock Weighted fairness notions for indivisible items revisited.
\newblock In {\em Proc.\ 36th Conf.\ Artif.\ Intell.\ (AAAI)}, pages
  4949--4956, 2022.

\bibitem{conitzerFS17}
Vincent Conitzer, Rupert Freeman, and Nisarg Shah.
\newblock Fair public decision making.
\newblock In {\em Proc.\ 18th Conf.\ Econ.\ Comput.\ (EC)}, pages 629--646,
  2017.

\bibitem{ConitzerF0V19}
Vincent Conitzer, Rupert Freeman, Nisarg Shah, and Jennifer~Wortman Vaughan.
\newblock Group fairness for the allocation of indivisible goods.
\newblock In {\em Proc.\ 33rd Conf.\ Artif.\ Intell.\ (AAAI)}, pages
  1853--1860, 2019.

\bibitem{freemanSV20}
Rupert Freeman, Nisarg Shah, and Rohit Vaish.
\newblock Best of both worlds: Ex-ante and ex-post fairness in resource
  allocation.
\newblock In {\em Proc.\ 21st Conf.\ Econ.\ Comput.\ (EC)}, pages 21--22, 2020.

\bibitem{GargKK20}
Jugal Garg, Pooja Kulkarni, and Rucha Kulkarni.
\newblock Approximating {N}ash social welfare under submodular valuations
  through (un)matchings.
\newblock In {\em Proc.\ 31st Symp.\ Discret.\ Algorithms (SODA)}, pages
  2673--2687, 2020.

\bibitem{liptonMMS04}
Richard~J Lipton, Evangelos Markakis, Elchanan Mossel, and Amin Saberi.
\newblock On approximately fair allocations of indivisible goods.
\newblock In {\em Proc.\ 5th Conf.\ Electr.\ Commerce (EC)}, pages 125--131,
  2004.

\bibitem{Orlin10}
James Orlin.
\newblock Improved algorithms for computing {F}isher's market clearing prices.
\newblock In {\em Proc.\ 42nd Symp.\ Theory Comput.\ (STOC)}, pages 291--300,
  2010.

\bibitem{haleviEtAl}
Erel Segal-Halevi and Balazs Sziklai.
\newblock Monotonicity and competitive equilibrium in cake-cutting.
\newblock {\em Economic Theory}, 68(2):363–401, 2019.

\bibitem{steinhaus}
Hugo Steinhaus.
\newblock Sets on which several measures agree.
\newblock {\em Econometrica}, 16:101--104, 1948.

\bibitem{stromquist}
Walter Stromquist and Douglas Woodall.
\newblock Sets on which several measures agree.
\newblock {\em J. Math.\ Anal.\ Appl.}, 108:241--248, 1985.

\bibitem{suksompongT22}
Warut Suksompong and Nicholas Teh.
\newblock On maximum weighted nash welfare for binary valuations.
\newblock {\em Math.\ Soc.\ Sci.}, 117:101--108, 2022.

\end{thebibliography}
\appendix

\section{Preliminaries}\label{appendix:preliminaries}

\begin{customtheorem}[\Cref{thm:pickingWEFxy}]
Let $t_i, t_j$ be the number of picks of agents $i$, $j$, respectively, in a prefix of $\pi$. A picking sequence $\pi$ is $\WEF(x,y)$ if and only if for every prefix of $\pi$ and every pair of agents $i,j$, we have 
$	\frac{t_i + y}{w_i}\geq \frac{t_j-x}{w_j}.$
\end{customtheorem}
\begin{proof}
For completeness, we report every step of the proof of Theorem 3.2 in Appendix A.2 of the full version of \cite{chakrabortySHS22}. We point out that our adapted proof requires only a minor modification in inductive proof of Eq.~\eqref{eq:claim}. 

($\Rightarrow$) 
Assume that $\pi$ fulfills $\WEF(x,y)$.
Since the $\WEF(x,y)$ condition must be satisfied for every instance, we can choose a special one that forces the utility of every agent to equal her number of picks up to a certain point. Consider a prefix of $\pi$. Every agent values each item which has been picked so far with 1, and the remaining ones with 0. 
If $t_j=0$, the claim trivially holds. Otherwise, $\WEF(x,y)$ gives us existence of $g\in A_j$ such that 
\begin{align*}
\frac{v_i(A_i) + y\cdot v_i(g)}{w_i} 
\geq
\frac{v_i(A_j) - x\cdot v_i(g)}{w_j}.
\end{align*}
Plugging in $v_i(g)=1$ and $v_i(A_i)=t_i$ as well as $v_i(A_j)=t_j$ yields the claim.

($\Leftarrow$)
Consider any two agents $i,j$.
We show that the $\WEF(x,y)$ condition for agent $i$ towards agent $j$ is fulfilled after every pick of $j$.
Consider the $t_j$-th pick of agent $j$.
We divide the sequence of picks up to this point into \emph{phases}, where each phase $\ell\in\{1,\dots, t_j\}$ consists of the picks after agent $j$'s $(\ell-1)$-th pick up to (and including) the agent's $\ell$-th pick. 
We use the following notation:
\begin{itemize}
	\item $\tau_\ell := $ the number of times agent $i$ picks in phase $\ell$ (that is, between agent $j$'s $(\ell-1)$-th and $\ell$-th picks),
	\item 
	$\alpha_\ell := $ the total utility gained by agent $i$ in phase $\ell$,
	\item 
	$\beta_\ell := $ agent $i$'s utility for the item that agent $j$ picks at the end of phase $\ell$.
\end{itemize}

Let $\rho := w_i/w_j$. 
For any integer $s\in[t_j]$,
applying the condition in the theorem statement to the picking sequence up to and including phase $s$, we have
\begin{equation}
\label{eq:WWEF-picks}
y+\sum_{\ell=1}^s \tau_\ell ~\ge~ \rho(s-x).
\end{equation}
Every time agent $i$ picks, she picks an item with the highest value for her. 
In particular, in each phase $\ell$, she picks $\tau_{\ell}$ items each of which gives at least as high value to her as each item picked by agent $j$ after (and including) phase $\ell$. Hence for all phases $\ell \in\{1,\dots,t_j\}$,
\begin{equation}
\label{eq:WWEF-utils}
\alpha_\ell
\geq \tau_\ell \cdot \max_{\ell\le r \le t_j} \beta_r.
\end{equation}

To show the claim, we prove the following inequality for all $s\in[t_j]$:

\begin{align}
y\cdot \max_{1\le r\le t_j}\beta_r &+ \sum_{\ell=1}^s
\alpha_{\ell} \ge \rho\left(\sum_{\ell=1}^s\beta_\ell - x\beta_1 \right) \notag \\
&+ \left(y+\sum_{\ell=1}^s\tau_\ell - \rho(s-x)\right)\max_{s\le r\le t_j}\beta_r. \label{eq:claim}
\end{align}
We prove \eqref{eq:claim} by induction on $s$.
The base case $s=1$ can be obtained by setting $\ell=1$ in \eqref{eq:WWEF-utils} and adding the term $y\cdot \max_{1\le r\le t_j}\beta_r$ on both sides:
\begin{align*}
y\cdot \max_{1\le r\le t_j} &\beta_r + \alpha_1
\ge (y+\tau_1)\cdot  \max_{1\le r\le t_j}\beta_r \\
&\ge \rho (1-x)\beta_1 + (y+\tau_1-\rho (1-x))\cdot \max_{1\le r\le t_j}\beta_r
\end{align*}
For the inductive step, assume that \eqref{eq:claim} holds for some $s-1$.
Using the inductive hypothesis (i.h.), we have

\begin{align*}
y\cdot \max_{1\le r\le t_j}\beta_r + \sum_{\ell=1}^s \alpha_{\ell} & = y\cdot \max_{1\le r\le t_j}\beta_r + \sum_{\ell=1}^{s-1}\alpha_{\ell} + \alpha_{s} \\
&\overset{\text{(i.h.)}}{\ge} \rho\left(\sum_{\ell=1}^{s-1}\beta_\ell - x\beta_1\right) + \left(y+\sum_{\ell=1}^{s-1}\tau_\ell - \rho(s-1-x)\right) \cdot\max_{s-1\le r\le t_j}\beta_r + \alpha_{s} \\
&\overset{\eqref{eq:WWEF-picks}}{\ge} \rho\left(\sum_{\ell=1}^{s-1}\beta_\ell - x\beta_1\right) + \left(y+\sum_{\ell=1}^{s-1}\tau_\ell - \rho(s-1-x)\right) \cdot\max_{s\le r\le t_j}\beta_r  + \alpha_s \\
&\overset{\eqref{eq:WWEF-utils}}{\ge} \rho\left(\sum_{\ell=1}^{s-1}\beta_\ell - x\beta_1\right) + \left(y+\sum_{\ell=1}^{s-1}\tau_\ell - \rho(s-1-x)\right)  \cdot\max_{s\le r\le t_j}\beta_r  + \tau_s\max_{s\le r\le t_j}\beta_r \\
&= \rho\left(\sum_{\ell=1}^{s-1}\beta_\ell - x\beta_1\right) + \left(y+\sum_{\ell=1}^s\tau_\ell - \rho(s-1-x)\right) \cdot\max_{s\le r\le t_j}\beta_r \\
&= \rho\left(\sum_{\ell=1}^{s-1}\beta_\ell - x\beta_1\right) + \rho\max_{s\le r\le t_j}\beta_r + \left(y+\sum_{\ell=1}^s\tau_\ell - \rho(s-x)\right)
\cdot \max_{s\le r\le t_j}\beta_r \\
&\ge \rho\left(\sum_{\ell=1}^{s-1}\beta_\ell - x\beta_1\right) + \rho\beta_s  + \left(y+\sum_{\ell=1}^s\tau_\ell - \rho(s-x)\right)\cdot \max_{s\le r\le t_j}\beta_r \\
&= \rho\left(\sum_{\ell=1}^{s}\beta_\ell - x\beta_1\right) + \left(y+\sum_{\ell=1}^s\tau_\ell - \rho(s-x)\right) \cdot \max_{s\le r\le t_j}\beta_r.
\end{align*}
This completes the induction and shows \eqref{eq:claim}.

Using \eqref{eq:claim} with $s = t_j$ together with \eqref{eq:WWEF-picks} yields
\begin{align*}
y\cdot \max_{1\le r\le t_j}\beta_r + \sum_{\ell=1}^{t_j}
\alpha_{\ell}
~\ge~ \rho\left(\sum_{\ell=1}^{t_j}\beta_\ell - x\beta_1 \right).
\end{align*}

Now letting $A_i$ and $A_j$ be the bundles of agents $i$ and $j$ after agent $j$'s $t_j$-th pick, and $g$ be an item in $A_j$ for which agent $i$ has highest utility, we obtain from the last inequality that
\begin{align*}
y\cdot v_i(g) + v_i(A_i) 
&\ge
\frac{w_i}{w_j}\cdot (v_i(A_j) - x\beta_1) \\
&\ge
\frac{w_i}{w_j}\cdot (v_i(A_j) - x\cdot v_i(g)).
\end{align*}
Therefore the $\WEF(x,y)$ condition for agent $i$ towards agent $j$ is fulfilled, completing the proof.

%
%
%
%
\end{proof}

\section{Additive Valuations}

\subsection{Ex-ante $\WEF$ and Ex-post $\WEF(1,1)$+$\WPROP1$}

\begin{customtheorem}[\Cref{prop:WSEEFimpliesWEFexante}]
Given a fractional allocation $X$, if $X$ is ex-ante $\WSD$-$\EF$, then $X$ is ex-ante $\WEF$. 
\end{customtheorem}
\begin{proof}
	Let $X$ be a fractional allocation satisfying $\WSD$-$\EF$. Let $i$ and $j$ be two distinct agents.
	
	First notice that, by $\WSD$-$\EF$, $a= w_j\cdot\sum_{g\in \goods}x_{ig}\geq w_i\cdot\sum_{g\in \goods}x_{jg}= b$.
	
	Let us consider $g_1, \dots, g_m$ sorted according to $i$'s valuations. 
	We define $f,h: [0, a]\rightarrow \mathbb{R}^+$ in the following way: 
	\begin{align*}
		f(x) = \begin{cases}
			v_i(g_1) & \text{ if } x = 0\\
			v_i(g_k)& \text{ if } x\in \left( y_{k-1}, y_k \right]
		\end{cases}
	\end{align*}
	where $y_k= w_j \cdot  \sum_{\ell=1}^{k} x_{ig_\ell} $ , and 
	\begin{align*}
		h(x) = \begin{cases}
			v_i(g_1) & \text{ if } x = 0\\
			0 & \text{ if } x\in (b, a] \\
			v_i(g_k) & \text{ if } x  \in \left( z_{k-1}, z_k \right], 
		\end{cases}
	\end{align*}
	where,  $z_k= w_i \cdot\sum_{\ell=1}^{k}  x_{jg_\ell}$.
	
	Notice that the difference between $y_k$ and $z_k$  is that the summations concern fractions of $g_1, \dots, g_k$ in $X_i$ and $X_j$, respectively. Moreover, the summations are multiplied by $w_j$ and $w_i$, respectively.
	
	We next show that $\WSD$-$\EF$ implies $f(x)\geq h(x)$, for all $x\in [0,a]$.
	
	Notice that for $x\in [b, a]$, $f(x)>0$ and $h(x)=0$. Assume towards a contradiction that there exists $x\in [0,b]$ such that $f(x)< h(x)$. Then, there exist $\ell_1$ and $\ell_2$ such that 
	\[
	v_i(g_{\ell_1})= f(x) < h (x) = v_i(g_{\ell_2}) \, .
	\]
	By definition of $f$ and $h$, we have 
	\[
	x \in (y_{\ell_1-1}, y_{\ell_1}]\cap (z_{\ell_2-1}, z_{\ell_2}] \, .
	\]
	Hence, $y_{\ell_1-1}< z_{\ell_2}$.
	
	Moreover, since $ v_i(g_{\ell_1 -1}) \leq v_i(g_{\ell_1}) < v_i(g_{\ell_2})$ and goods sorted in a non-increasing ordering with respect to $i$'s valuations, 
	\[
	\set{g_1, \dots, g_{\ell_2}} \subseteq \set{g \ : \ v_i(g)\geq v_i(g_{\ell_2})}\subset\set{g_1, \dots, g_{\ell_1}} \, .
	\]
	These imply
	\begin{equation*}
		y_{\ell_1-1} = w_j \cdot  \sum_{\ell=1}^{\ell_1} x_{ig_\ell} \geq w_j \cdot\sum_{g \in \set{g \ : \ v_i(g)\geq v_i(g_{\ell_2})}}  x_{ig_\ell}
	\end{equation*}
	and 
	\begin{equation*}
		z_{\ell_2} = w_i \cdot \sum_{\ell=1}^{\ell_2} x_{jg_\ell} \leq w_i \cdot \sum_{g \in \set{g \ : \ v_i(g)\geq v_i(g_{\ell_2})}}  x_{jg_\ell} \ .
	\end{equation*}
	
	These two inequalities together with $y_{\ell_1-1}< z_{\ell_2}$, show
	$$ w_j \cdot \sum\limits_{g \in  \set{g \ : \ v_i(g)\geq v_i(g_{\ell_2})}}  x_{ig_\ell} < w_i \cdot\sum\limits_{g \in  \set{g \ : \ v_i(g)\geq v_i(g_{\ell_2})}}  x_{jg_\ell}$$
	-- a contradiction with $\WSD$-$\EF$.
	
	To conclude, we notice that 
	\[
	w_j\cdot \expectation{v_i(X_i)} = \int_0^a f(x) dx
	\]
	and
	\[ 
	w_i\cdot \expectation{v_i(X_j)} = \int_0^a h(x) dx \, .
	\]
	 Since $f(x)\geq h(x)$, for all $x\in [0,a]$, $w_j\cdot \expectation{v_i(X_i)}\geq w_i\cdot \expectation{v_i(X_j)}$ follows.
\end{proof}

\section{General Valuations and Equal Entitlements}
\subsection{Cancelable Valuations}\label{appendix:cancelable}
As mentioned above, for additive valuations and equal entitlements the $\UG$-decomposition of $\XDSE$ is ex-ante $\SD$-$\EF$ as well as ex-post $\EF1$. The latter results from the fact that every allocation in the support of the lottery emerges from an RB picking sequence. We now describe what implications this approach has on cancelable valuations functions. Given the cancelable valuation function  $v_i$, for each $i\in \agents$, we create a corresponding additive valuation function $\hat{v}_i$ with $\hat{v}_i(g)=v_i(g)$. Then we apply the result for additive valuations and equal entitlements. This implies the obtained lottery is ex-ante $\SD$-$\EF$, as the priorities over goods are the same for $v_i$ and $\hat{v}_i$.

Let us discuss ex-post properties. We show that any RB picking sequence for $\hat{v}$ yields an allocation that is ex-post $\EF1$ for $v$. To this aim, we need the following lemmas.

\begin{lemma}\label{lemma:cancelableProperty1}
	Let $v$ be a cancelable valuation function.
	Given any $S,T\subseteq \goods$ and any $R \subseteq \goods \setminus (S \cup T)$, if $v(S) \geq v(T)$, then 
	\[
	v(S \cup R) \ge v(T \cup R).
	\]
\end{lemma}

\begin{proof}
	Let $R_i = \{g_1, \ldots, g_i\}$ and $R_k= R$. If $v(T \cup R_i) > v(S \cup R_i)$,
	then $v(T \cup R_{i-1}) > v(S \cup R_{i-1})$ by cancelability of $v$.
	Applying this result for all $i$ from $k$ down to $1$, we get
	that if $v(T \cup R) > v(S \cup R)$, then $v(T) > v(S)$ -- a contradiction.
\end{proof}

\begin{lemma}\label{lemma:cancelableProperty2}
	Let $v$ be a cancelable valuation function.
	Given any $S,T\subseteq \goods$ and any $g,g' \in \goods \setminus (S\cup T)$, if $v(S) \geq v(T)$ and $v(g)\geq v(g')$, then 
	\[
	v(S\cup\set{g}) \geq v(T\cup\set{g'}).
	\]
\end{lemma}
\begin{proof}
	Applying \Cref{lemma:cancelableProperty1} twice, $v(S\cup\set{g}) \geq v(T\cup\set{g}) \geq  v(T\cup\set{g'})$ follows.
\end{proof}

Recall that an RB picking sequence takes as input an ordering of the agents and agents' priorities over goods.
Consequently, an RB picking sequence produces the very same outcome for $v$ and $\hat{v}$ (if ties are broken in the same manner).

\begin{proposition}
	Given a fair division instance with cancelable valuations, any RB picking sequence yields an allocation that is $\EF1$.
\end{proposition}
\begin{proof}
	Let $Y$ be an allocation obtained by an RB picking sequence. 
	
	Given any pair of agents $i,j$, the sizes of $Y_i$ and $Y_j$ differ by at most $1$. We next show that agent $i$ is $\EF1$ towards agent $j$: Without loss of generality, we can assume $\modulus{Y_j}= \modulus{Y_i}+1$. Let us denote by $g^*$ the most preferred good of $i$ in $Y_j$. Because of the RB property, there exists a matching $M=\set{(g,g') \vert g\in Y_i, g'\in Y_j\setminus\set{g^*}}$ such that $v_i(g)\geq v_i(g')$ for each $(g,g')\in M$. Recursively applying~\Cref{lemma:cancelableProperty2} over the pairs in $M$ we obtain $v_i(Y_i)\geq v_i( Y_j\setminus\set{g^*})$. Therefore, $i$ is $\EF1$ towards agent $j$; applying the same arguments on any pair of agents the thesis follows.
\end{proof}

Overall, our result for cancelable valuations is as follows. 

\begin{theorem}
	For equal entitlements and cancelable valuations, the $\UG$-decomposition of $\XDSE$ is ex-ante $\SD$-$\EF$ and ex-post $\EF1$.
\end{theorem}

It remains to explore whether ex-ante $\SD$-$\EF$ implies ex-ante $\EF$ for cancelable valuations.

\end{document}